\documentclass[submission,copyright,creativecommons]{eptcs}
\providecommand{\event}{ICE 2011} 
\usepackage{breakurl}             

\title{Innocent strategies as presheaves\\ and interactive
  equivalences for CCS}%
 \author{Tom Hirschowitz \thanks{Both authors
    have been partially funded by the French projects CHoCo (ANR-07-BLAN-0324), PiCoq (ANR 2010 BLAN 0305
    01), and CNRS PEPS CoGIP.}  \institute{CNRS, Universit\'e de Savoie
    \\ Chamb\'ery, France}
  \and
  Damien Pous
  \institute{CNRS, Laboratoire d'Informatique de Grenoble\\
    Grenoble, France}
}
\def\titlerunning{Innocent strategies as presheaves \ldots}
\def\authorrunning{T.\ Hirschowitz \& D.\ Pous}
\usepackage[utf8]{inputenc}
\usepackage{
mathpartir,
amsmath,amsfonts,amsthm,amssymb,%
stmaryrd,
graphicx,empheq,wrapfig
} %
\usepackage[mathcal]{euscript}

\DeclareMathOperator{\ob}{ob}
\DeclareMathOperator{\mor}{mor}

\DeclareMathOperator{\dom}{dom}
\DeclareMathOperator{\cod}{cod}

\DeclareMathOperator{\Fam}{Fam}

\newcommand{\bigcat}[1]{\ensuremath{\mathsf{#1}}}
\newcommand{\maji}[1]{\ensuremath{\mathbb{#1}}}

\newcommand{\Vars}{\bigcat{Vars}}
\newcommand{\Names}{\bigcat{Names}}

\newcommand{\commentthis}[1]{}

\newcommand{\Cospan}[1]{\bigcat{Cospan} (#1)}

\newcommand{\into}{\hookrightarrow}
\newcommand{\otni}{\hookleftarrow}

\newcommand{\xto}[1]{\xrightarrow{#1}}
\newcommand{\xot}[1]{\xleftarrow{#1}}

\renewcommand{\hat}[1]{\widehat{#1}}

\newcommand{\B}{\maji{B}}

\newcommand{\C}{\maji{C}}

\newcommand{\tick}{\daimon}

\newcommand{\Chat}{\hat{\maji{C}}}

\newcommand{\CE}{\E}

\newcommand{\CV}{\V}

\newcommand{\CVX}{\CV_X}
\newcommand{\CVY}{\CV_Y}

\newcommand{\CEXhat}{\widehat{\CE_X}}

\newcommand{\CVnhat}{\widehat{\CV_n}}
\newcommand{\CVXhat}{\widehat{\CVX}}

\newcommand{\CVYhat}{\widehat{\CVY}}

\newcommand{\CEX}{\CE_X}

\newcommand{\D}{\maji{D}}

\newcommand{\E}{\maji{E}}

\newcommand{\Eh}{\B}

\newcommand{\F}{\maji{F}}

\newcommand{\G}{\maji{G}}
\newcommand{\GG}{\bigcat{G}}

\newcommand{\MMM}{\mathcal{M}}

\newcommand{\V}{\maji{V}}

\newcommand{\W}{\maji{W}}

\newcommand{\CW}{\W}

\newcommand{\CWofX}{\CW(X)}

\newcommand{\CWofXhat}{\widehat{\CWofX}}

\newcommand{\X}{\maji{X}}

\renewcommand{\SS}{\bigcat{S}}

\newcommand{\SSX}{\SS_X}

\newcommand{\Set}{\bigcat{Set}}

\newcommand{\FinOrd}{\bigcat{FinOrd}}
\newcommand{\OPsh}[1]{{\widehat{#1}}^f}

\newcommand{\Cat}{\bigcat{Cat}}

\newcommand{\CAT}{\bigcat{CAT}}

\newcommand{\trou}{\boxempty}

\newcommand{\un}{I}
\newcommand{\para}{\mathbin{\mid}}

\newcommand{\transl}[1]{\llbracket #1 \rrbracket}
\newcommand{\rond}{\circ}

\newcommand{\id}{\mathit{id}}
\newcommand{\iso}{\cong}
\newcommand{\impll}{\multimap}
\newcommand{\tens}{\otimes}

\newcommand{\card}[1]{|#1|}

\newcommand{\eq}{\mathrel{\texttt{:=}}}
\newcommand{\ens}[1]{\{ #1 \}}

\newcommand{\aalt}{\mathrel{|}}

\newcommand{\op}[1]{#1^{\mathit{op}}}

\newcommand{\subs}[1]{[#1]}

\newcommand{\Gam}{\Gamma}
\newcommand{\Del}{\Delta}

\newcommand{\equi}{\simeq}

\newcommand{\bureaucratic}[1]{}

\newcommand{\abar}{\overline{a}}

\newcommand{\sender}{\epsilon}
\newcommand{\receiver}{\rho}

\newcommand{\iotapos}[1]{\iota^+_{#1}}
\newcommand{\iotaneg}[1]{\iota^-_{#1}}
\newcommand{\iotaposni}{\iotapos{n,i}}

\newcommand{\iotaposnj}{\iotapos{n,j}}
\newcommand{\iotanegni}{\iotaneg{n,i}}
\newcommand{\iotanegmj}{\iotaneg{m,j}}
\newcommand{\iotanegnj}{\iotaneg{n,j}}
\newcommand{\paraof}[1]{\pi_{#1}}
\newcommand{\paralof}[1]{\pi^l_{#1}}
\newcommand{\pararof}[1]{\pi^r_{#1}}
\newcommand{\paran}{\paraof{n}}
\newcommand{\paraln}{\paralof{n}}
\newcommand{\pararn}{\pararof{n}}

\newcommand{\nuof}[1]{\nu_{#1}}
\newcommand{\nun}{\nuof{n}}
\newcommand{\tickof}[1]{\tick_{#1}}
\newcommand{\tickn}{\tickof{n}}
\newcommand{\tauof}[4]{\tau_{#1,#2,#3,#4}}
\newcommand{\taunimj}{\tauof{n}{i}{m}{j}}

\newcommand{\Pl}{\mathrm{Pl}}

\newcommand\independent{\protect\mathpalette{\protect\independenT}{\perp}} 
\def\independenT#1#2{\mathrel{\rlap{$#1#2$}\mkern2mu{#1#2}}} 
\newcommand{\bbot}{\mathord{\independent}}
\newcommand{\daimon}{\heartsuit}

\newcommand{\recin}[2]{\texttt{rec}\ #1\ \texttt{in}\ #2}
\newcommand{\Gl}{\GG}
\newcommand{\GlX}{\Gl_X}
\newcommand{\glue}{\parallel}
\newcommand{\SStoGG}{\bigcat{Gl}}

\newcommand{\state}{\sigma}
\newcommand{\statei}{\sigma'}

\newcommand{\floor}[1]{\lfloor #1 \rfloor}

\DeclareMathOperator{\Ran}{Ran}

\def\framed{%
\setbox0=\vbox\bgroup%
\advance\hsize by -2\fboxsep\advance\hsize by -2\fboxrule%
\linewidth=\hsize%
}
\def\endframed{%
\egroup\noindent\framebox[\textwidth]{\box0}\vspace*{1mm}}

\usepackage{tikz,net}
\tikzset{todim/.style = {decoration={markings, mark=at position .5 with %
      {\draw (-1pt,-1pt) rectangle (1pt,1pt);}},postaction={decorate}}}
\renewcommand{\diaggrandhauteur}{.6}
\renewcommand{\diaggrandlargeur}{1.3}

\newtheorem{defn}{Definition}
\newtheorem{lem}{Lemma}
\newtheorem{thm}{Theorem}
\newtheorem{prop}{Proposition}
\newtheorem{ex}{Example}
\newtheorem{rk}{Remark}

\begin{document}
\maketitle
\begin{abstract}
  Seeking a general framework for reasoning about and comparing
  programming languages, we derive a new view of Milner's
  CCS~\cite{Milner80}.  We construct a category $\CE$ of \emph{plays},
  and a subcategory $\CV$ of \emph{views}.  We argue that presheaves
  on $\CV$ adequately represent \emph{innocent} strategies, in the
  sense of game semantics~\cite{hyland97}. We then equip innocent
  strategies with a simple notion of interaction.  This results in an
  interpretation of CCS.

  Based on this, we propose a notion of \emph{interactive equivalence}
  for innocent strategies, which is close in spirit to Beffara's
  interpretation~\cite{bef05:lrc} of testing
  equivalences~\cite{DBLP:journals/tcs/NicolaH84} in concurrency
  theory.  In this framework we prove that the analogues of
  \emph{fair} and \emph{must} testing equivalences coincide, while
  they differ in the standard setting.
\end{abstract}
\section{Overview}

\paragraph{Theories of programming languages}
Research in programming languages is mainly technological. Indeed, it
heavily relies on techniques which are ubiquitous in the field, but
almost never formally made systematic. Typically, the definition of a
language then quotiented by variable renaming ($\alpha$-conversion)
appears in many theoretical papers about functional programming
languages. Why isn't there yet any abstract framework performing these
systematic steps for you?  Because the quest for a real \emph{theory}
of programming languages is not achieved yet, in the sense of a corpus
of results that actually help developing them or reasoning about them.
However, many attempts at such a theory do exist.

A problem for most of them is that they do not account for the
dynamics of execution, which limits their range of application. This
is for example the case of Fiore et al.'s second-order
theories~\cite{DBLP:conf/lics/Fiore08,hirscho:lam,DBLP:journals/iandc/HirschowitzM10}.
A problem for most of the other theories of programming languages is
that they neglect denotational semantics, i.e., they do not provide a
notion of model for a given language.  This is for example the case of
Milner et al.'s \emph{bigraphs}~\cite{Milner:bigraphs}, or of most
approaches to \emph{structural operational
  semantics}~\cite{PlotkinSOS}, with the notable exception of the
\emph{bialgebraic semantics} of Turi and
Plotkin~\cite{plotkin:turi:bialgebraic}.  A recent, related, and
promising approach is \emph{Kleene coalgebra}, as advocated by
Bonsangue et al.~\cite{DBLP:conf/fossacs/BonsangueRS09}.  Finally,
\emph{higher-order rewriting}~\cite{DBLP:conf/lics/Nipkow91}, and its
semantics in double categories~\cite{DBLP:conf/birthday/GadducciM00}
or in cartesian closed
2-categories~\cite{HIRSCHOWITZ:2010:HAL-00540205:2}, is not currently
known to adequately account for process calculi.

\paragraph{Towards a new approach}
The most relevant approaches to us are bialgebraic semantics and
Kleene coalgebra, since the programme underlying the present paper
concerns a possible alternative.
A first difference, which is a bit technical but may be of importance,
is that both bialgebraic semantics and Kleene coalgebra are based on
labelled transition systems (LTSs), while our approach is based on
reduction semantics. Reduction semantics is often considered more
primitive than LTSs, and much work has been devoted to deriving the
latter from the
former~\cite{DBLP:conf/concur/Sewell98,Milner:bigraphs,Sobocinski:grpos,modularLTS}.
It might thus be relevant to propose a model based only on the more
primitve reduction semantics.

More generally, our approach puts more emphasis on interaction between
programs, and hence is less interesting in cases where there is no
interaction. A sort of wild hope is that this might lead to unexpected
models of programming languages, e.g., physical ones. This could also
involve finding a good notion of morphism between languages, and
possibly propose a notion of compilation.  At any rate, the framework
is not set up yet, so investigating the precise relationship with
bialgebraic semantics and Kleene coalgebra is deferred to further
work.

How will this new approach look like? Compared to
such long-term goals, we only take a small step forward here, by
 considering a particular case, namely Milner's
CCS~\cite{Milner80}, and providing a new view of it.  This view
borrows ideas from the following lines of research: game
semantics~\cite{hyland97}, and in particular the notion of an
\emph{innocent strategy} and \emph{graphical
  games}~\cite{Miller08,DBLP:journals/entcs/HirschowitzHH09}, Krivine
realisability~\cite{DBLP:journals/tcs/Krivine03},
ludics~\cite{DBLP:journals/mscs/Girard01}, testing equivalences in
concurrency~\cite{DBLP:journals/tcs/NicolaH84,bef05:lrc}, the presheaf
approach to
concurrency~\cite{DBLP:conf/lics/JoyalNW93,DBLP:journals/mscs/KasangianL99},
and sheaves~\cite{MM}. But it is also related to, e.g., graph
rewriting~\cite{Ehrig}, 
and
computads~\cite{DBLP:journals/tcs/Burroni93}.

\paragraph{From games to presheaves} Game semantics~\cite{hyland97}
has provided fully complete models of programming languages. However,
it is based on the notion of a \emph{strategy}, i.e., a set of
\emph{plays} in the game, satisfying a few conditions. In concurrency
theory, taking as a semantics the set of accepted plays, or `traces',
is known as \emph{trace semantics}. Trace semantics is generally
considered too coarse, since it equates, for a most famous example,
the right and the wrong coffee machines, $a.(b+c)$ and $ab +
ac$~\cite{Milner80}.

An observation essentially due to Joyal, Nielsen, and Winskel is that
strategies, i.e., prefix-closed sets of plays, are actually particular
\emph{presheaves of booleans} on the category $\C$ with plays as
objects, and prefix inclusions as morphisms. By presheaves of booleans
on $\C$ we here mean functors $\op\C \to 2$, where $2$ is the preorder
category $0 \leq 1$. If a play $p$ is \emph{accepted}, i.e., mapped to
$1$, then its prefix inclusions $q \into p$ are mapped to the unique
morphism with domain $1$, i.e., $\id_1$, which entails that $q$ is
also accepted.

We consider instead presheaves (of sets) on $\C$. So, a play $p$ is
now mapped to a set $S(p)$, to be thought of as the set of ways for
$p$ to be accepted by the strategy $S$. Considering the set of players
as a team, $S(p)$ may also be thought of as the set of \emph{states}
of the team after playing $p$.

Presheaves are fine enough to account for
bisimilarity~\cite{DBLP:conf/lics/JoyalNW93,DBLP:journals/mscs/KasangianL99}.
Indeed, they are essentially forests with edges labelled in moves. For
example, in the setting where plays are finite words on an alphabet,
the wrong coffee machine may be represented by the presheaf $S$
defined by the equations on the left and pictured as on the right:
\begin{center}
  \begin{center}
    \begin{minipage}[c]{0.2\linewidth}
      \begin{itemize}
      \item $S (\epsilon) = \ens{\star}$,
      \item $S (a) = \ens{x,x'}$,
      \item $S (ab) = \ens{y}$,
      \item $S (ac) = \ens{y'}$,
      \end{itemize}
    \end{minipage}
      \hfil
    \begin{minipage}[c]{0.35\linewidth}
      \begin{itemize}
      \item $S (\epsilon \into a) = \ens{x \mapsto \star, x' \mapsto \star}$,
      \item $S (a \into ab) = \ens{y \mapsto x}$,
      \item $S (a \into ac) = \ens{y' \mapsto x'}$:
      \end{itemize}
    \end{minipage}
      \hfil
    \diag(.15,.8){%
      \& |(root)| \star \& \\
      |(x)| x \& \& |(x')| x' \\
      \\
      |(y)| y \& \& |(y')| y'. %
    }{%
      (root) edge[-,labelal={a}] (x) %
      edge[-,labelar={a}] (x') %
      (x)  edge[-,labell={b}] (y) %
      (x')  edge[-,labelr={c}] (y') %
    }
  \end{center}
\end{center}
So, in summary: the standard notion of strategy may be generalised to
account for branching equivalences, by passing from presheaves of
booleans to presheaves of sets.

\paragraph{Multiple players}
Traditional game semantics mostly emphasises two-player
games. There is an implicit appearance of three-player games in the
definition of composition of strategies, and of four-player games in
the proof of its associativity, but these games are never given a
proper status.  A central idea of graphical games, and to a lesser
extent of ludics, is the emphasis on multiple-player games.

Here, there first is a base category $\B$ of \emph{positions}, whose
objects represent configurations of players to which the game may
arrive at.  Since the game represents CCS, it should be natural that
players are related to each other via the knowledge of
\emph{communication channels}.  So, positions are bipartite graphs
with vertex sets \emph{players} and \emph{channels}, and edges from
channels to players indicating when the former is known to the latter.
As a first approximation, morphisms of positions may be thought of as
just embeddings of such graphs.

Second, there is a category $\E$ of \emph{plays}, with a functor to
$\B$ sending each play to its initial position. Plays are represented
in a more flexible way than just sequences of moves, namely using a
kind of string diagrams. This echoes the idea~\cite{Mellies04} that
two moves may be independent, and that plays should not depend on the
order in which two independent moves are performed. Furthermore, our
plays are a rather general notion, allowing, e.g., to look at how only
some players of the initial position evolve. Morphisms of plays
account both for:
\begin{itemize}
\item prefix inclusion, i.e., inclusion of a play into a longer play, and
\item position enlargement, e.g., inclusion of information about some
  players into information about more players.
\end{itemize}

Now, restricting to plays above a given position $X$, and then taking
presheaves on this category $\E_X$, we have a category of strategies
on $X$.

\paragraph{Innocence}
A fundamental idea of game semantics is the notion of
\emph{innocence}, which says that players have a restricted
\emph{view} of the play, and that their actions may only depend on
that view. 

We implement this here by defining a subcategory $\CVX \into \CEX$ of
\emph{views} on $X$, and deeming a presheaf $F$ on $\E_X$
\emph{innocent} when it is determined by its restriction $F'$ to
$\CVX$, in the sense that it is isomorphic to the \emph{right Kan
  extension}~\cite{MacLane:cwm} of $F'$ along $\op{\CVX} \into \op{\CEX}$.

Given this, it is sensible to define innocent strategies to be just
presheaves on $\CVX$, and view them as strategies via the (essential)
embedding $\CVXhat \into \CEXhat$ induced by right Kan extension.

\paragraph{Interaction}
For each position $X$, we thus have a category $\SS_X = \CVXhat$ of
innocent strategies.  In game semantics, composition of strategies is
achieved in two steps: \emph{interaction} and \emph{hiding}.
Essentially, interaction amounts to considering the three-player game
obtained by letting two two-player games interact at a common
interface. Hiding then forgets what happens at that interface, to
recover a proper two-player game.

We have not yet investigated hiding in our approach, but, thanks to
the central status of multiple-player games, interaction is accounted
for in a very streamlined way.  For any position $X$ with two
subpositions $X_1 \into X$ and $X_2 \into X$ such that each player is
in either $X_1$ or $X_2$, but none is in both, given strategies $F_1
\in \SS_{X_1}$ and $F_2 \in \SS_{X_2}$, there is a unique innocent
strategy (up to canonical isomorphism in $\SS_X$), the
\emph{amalgamation} $[F_1, F_2]$ of $F_1$ and $F_2$, whose
restrictions to $X_1$ and $X_2$ are $F_1$ and $F_2$ (again up to
isomorphism).  

Amalgamation in this sense models interaction in the sense of game
semantics, and, using the correspondence with presheaves on $\CEX$
given by right Kan extension, it is the key to defining interactive
equivalences.

\paragraph{CCS}
Next, we define a translation of CCS terms with recursive equations
into innocent strategies. This rests on \emph{spatial} and
\emph{temporal} decomposition results, which entail that innocent
strategies are a solution of a system of equations of categories (up
to equivalence). A natural question is then: which equivalence does
this translation induce on CCS terms ?

\paragraph{Interactive equivalences}
Returning to the development of our approach, we then define a notion
of \emph{interactive equivalence}, which is close in spirit to both
testing equivalences in concurrency theory and Krivine realisability
and ludics.

The game, as sketched above, allows interacting with players which are
not part of the considered position. E.g., a player in the considered
position $X$ may perform an input which is not part of any
synchronisation. A \emph{test} for a strategy $F$ on $X$ is then,
roughly, a strategy $G$ on a position $X'$ with the same names as $X$.
To decide whether $F$ \emph{passes} the test $G$, we consider a
restricted variant of the game on the `union' $X \cup X'$, forbidding
any interaction with the outside. We call that variant the
\emph{closed-world} game. Then $F$ passes $G$ iff the amalgamation
$[F, G]$, right Kan extended to $\CE_{X \cup X'}$ and then restricted
to the closed-world game, belongs to some initially fixed class of
strategies, $\bbot_{X \cup X'}$. Finally, two strategies $F$ and $F'$
on $X$ are equivalent when they pass the same tests.

Examples of $\bbot$ include:
\begin{itemize}
\item $\bbot^m$, consisting of all strategies whose maximal states
  (those that admit no strict extensions) all play a \emph{tick} move,
  fixed in advance; the tick move plays a r\^ole analogous to the
  daimon in ludics: it is the only move which is observable from the
  outside;
\item $\bbot^f$, consisting of all strategies in which all states
  admit an extension playing tick.
\end{itemize}
From the classical concurrency theory point of view on behavioural
equivalences, the first choice mimicks \emph{must} testing
equivalence, while the second mimicks \emph{fair} testing
equivalence~\cite{DBLP:journals/tcs/NicolaH84}.
%
%
As a somehow surprising result, we prove that $\bbot^f$ and $\bbot^m$
yield the same equivalence. The reason is that our notion of play is
more flexible than just sequences of moves, as we explain in
Section~\ref{subsec:fairmust}.

\paragraph{Summary}
In summary, our approach emphasises a flexible notion of
multiple-player play, encompassing both views in the sense of game
semantics, closed-world plays, and intermediate notions.  Strategies
are then described as presheaves on plays, while innocent strategies
are presheaves on views. Innocent strategies admit a notion of
interaction, or amalgamation, and are embedded into strategies via
right Kan extension. This allows a notion of testing, or interactive
equivalence by amalgamation with the test, right Kan extension, and
finally restriction to closed-world.

Our main technical contributions are then a translation of CCS terms
with recursive equations into innocent strategies, and the result that
fair and must equivalence coincide in our setting.

\paragraph{Perspectives}
So, we have defined a flexible category of multiple-player play,
combining inclusion in time (more moves) and in width (more players).
Having isolated a subcategory of views, we have defined innocent
strategies as presheaves on views, relative to a base position. We
have then translated CCS processes with recursive definitions into
innocent strategies. Then, using right Kan extension and restriction
to so-called closed-world plays, we have defined a notion of
interactive equivalence. Finally, we have proved that two interactive
equivalences, fair and must testing, coincide.

Our next task is clearly to tighten the link with CCS. Namely, we
should explore which equivalence on CCS is induced via our
translation, for a given interactive equivalence. We will start with
$\bbot^m$.  Furthermore, the very notion of interactive equivalence
might deserve closer consideration. Its current form is rather
\emph{ad hoc}, and one could hope to see it emerge more naturally from
the game. For instance, the fixed class $\bbot$ of `successful'
strategies should probably be subjected to more constraints than is
done here, but two examples were not enough to make any guess.  Also,
the paradigm of observing via the set of successful tests might admit
sensible refinements, e.g., probabilistic ones.

Another possible research direction is to tighten the link with
`graphical' approaches to rewriting, such as graph rewriting or
computads.  E.g., our plays might be presented by a
computad~\cite{Guiraud}, or be the bicategory of rewrite sequences up
to shift equivalence, generated by a graph grammar in the sense of
Gadducci et al.~\cite{DBLP:journals/tcs/GadducciHL99}. Both goals
might require some technical adjustments, however.  For computads, we
would need the usual yoga of U-turns to flexibly model our positions;
however, e.g., zigzags of U-turns are usually only equal up to a
higher-dimensional cell, while they would map to equal positions in
our setting.  For graph rewriting, the problem is that our positions
are not exactly graphs (e.g., the channels known to a player are
totally ordered).

Other perspectives include the treatment of more complicated calculi
like $\pi$ or $\lambda$. In particular, calculi with duplication of
terms will pose a serious challenge.  An even longer-term hope is to
be able to abstract over our approach. Is it possible to systematise
the process starting from a calculus as studied in programming
language theory, and generating its strategies modulo interactive
equivalence?  If this is ever understood, the next question is: when
does a translation between two such calculi preserve a given
interactive equivalence?  Finding general criteria for this might have
useful implications in programming languages, especially compilation.

\paragraph{Notation}
\begin{wrapfigure}{r}{0pt}
  \diag(.4,1){%
    |(FC)| FC \& |(FC')| FC' \\ %
    |(GD)| GD \& |(GD')| GD' %
  }{ %
    (FC) edge[labelu={F (f)}] (FC') %
    edge[labell={u}] (GD) %
    (FC') edge[labelr={u'}] (GD') %
    (GD) edge[labeld={G (g)}] (GD') %
    }
\end{wrapfigure}
The various categories and functors used in the development are summed
up with a short description in Table~\ref{tab:summary}. There, given
two functors $\C \xto{F} \E \xot{G} \D$, we denote (slightly
abusively) by $\C \downarrow_\E \D$ the \emph{comma} category: it has
as objects triples $(C,D,u)$ with $C \in \C$, $D \in \D$, and $u
\colon F (C) \to G (D)$ in $\E$, and as morphisms $(C,D,u) \to
(C',D',u')$ pairs $(f,g)$ making the square on the right
commute. Also, when $F$ is the identity on $\C$ and $G \colon 1 \to
\C$ is an object $C$ of $\C$, this yields the usual \emph{slice}
category, which we abbreviate as $\C / C$. Finally, the category of
presheaves on any category $\C$ is denoted by $\Chat = \Set^{\op\C}$.

\begin{table}[t]
  \centering
  \begin{tabular}[t]{|c|c|}
    \hline
    Category & Description of its objects \\ \hline
    $\Chat$ &  `diagrams'  \\
    $\B \into \Chat$ &  positions \\
    $\CE \into (\B \downarrow_{\Chat} \Chat) $ & plays \\
    $\CEX = (\CE \downarrow_{\B} (\B/X))$ & plays on a position $X$ \\
    $\CVX \into \CEX$ & views on $X$ \\
    $\SS_X = \CVXhat$ & innocent strategies on $X$ \\ 
    $\CW \into \CE$ & closed-world plays \\
    $\CWofX$ & closed-world plays on $X$ \\
    \hline
  \end{tabular}
  \caption{Summary of categories and functors}
  \label{tab:summary}
\end{table}

Furthermore, we denote, for any category $\C$, by $\ob (\C)$ its set
of obects and by $\mor (\C)$ its set of morphisms.  For any functor $F
\colon \C \to \D$, we denote by $\op\F \colon \op\C \to \op \D$ the
functor induced on opposite categories, defined exactly as $F$ on both
objects and morphisms. Also, recall that an \emph{embedding} of
categories is an injective-on-objects, faithful functor. This admits
the following generalisation: a functor $F \colon \C \to \D$ is
\emph{essentially injective on objects} when $FC \iso FC'$ implies $C
\iso C'$. Any faithful, essentially injective on objects functor is
called an \emph{essential embedding}.

\section{Plays as string diagrams}
We now describe our approach more precisely, starting with
multiple-player plays.  We remain at a not completely formal level,
especially for presenting plays, because our experience is that most
readers get stuck on that. However, the interested reader may have a
look at the formal definition in Appendix~\ref{sec:diagrams}.

\subsection{Positions}
\begin{wrapfigure}{r}{0pt}
      \diagramme[stringdiag={.6}{.6}]{}{%
        \path[-,draw] %
        (a) edge (j1) edge (j2) %
        (c) edge (j2) edge (j3) %
        (b) edge (j1) edge (j2) edge (j3) %
        ; %
        \node[diagnode,at= (j3.south east)] {\ \ \ } ; %
      }{%
        \& \canal{a}     \& \&  \canal{c} \\
    \joueur{j1} \& \& \joueur{j2} \& \& \joueur{j3} \\
    \& \& \canal{b} \&    
  }{%
  }%
\end{wrapfigure}
Since the game represents CCS, it should be natural that players are
related to each other via the knowledge of \emph{communication
  channels}. This is represented by a kind of\footnote{Only `a kind
  of', because, as mentioned above, the channels known to a player are
  linearly ordered.}  bipartite graph: an example position is on the
right.  Bullets represent players, circles represent channels, and
edges indicate when a player knows a channel. The channels known by a
player are linearly ordered, e.g., counterclockwise, starting from the
lower one.  Formally, as explained in Appendix~\ref{sec:diagrams},
positions are presheaves over a certain category $\C_1$.  Morphisms of
positions are natural transformations, which are roughly morphisms of
graphs, mapping players to players and channels to channels. In full
generality, morphisms thus do not have to be injective.  However, let
us restrict to injective morphisms for this expository
paper. Positions and morphisms between them form a category $\B$.

\subsection{Moves}
\emph{Plays} are then defined as glueings of \emph{moves} derived from
the very definition of CCS, which we now sketch. 
Moves come in three layers:
\begin{itemize}
\item \emph{basic} moves, which are used to define views below,
\item \emph{full} moves, which are used in the statement of temporal
  decomposition (Theorem~\ref{thm:temp}),
\item and \emph{closed-world} moves, which are used to define
  closed-world plays (which in turn are central to the notion of
  interactive equivalence).
\end{itemize}
It might here be worth providing some intuition on the difference
between the three notions of move. A closed-world move roughly
consists of some players (one or more) synchronising together in some
specified way, each of them forking into several `avatars'. A full
move gathers what concerns one of the players involved in such a
synchronisation.  A basic move is what one of its avatars remembers of
the move.

Let us start with a closed-world move which concerns only one player,
and which is hence also a full move: \emph{forking} In the case of a
player knowing two channels, the forking move is represented by the
diagram $P$:
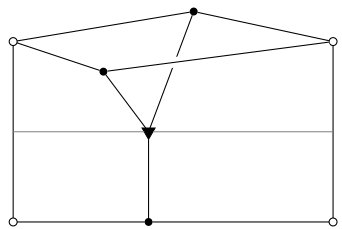
    \begin{equation}
  \diagramme[stringdiag={.4}{.6}]{}{
    \node[diagnode,at= (s1.south east)] {\ \ \ ,} ; %
  }{%
    \& \& \&  \& \joueur{t_2} \\
    \canal{t0} \& \& \& \& \& \& \& \canal{t1} \& \& \\ 
    \& \& \joueur{t_1}  \&  \\
    \& \ \& \\
    \coord{i0} \& \& \& \couppara{para} \& \& \& \& \coord{i1} \& \& \\ 
    \& \ \& \\
    \& \&  \\
    \canal{s0} \& \& \& \joueur{s} \& \& \& \& \canal{s1} \& \& 
  }{%
    (para) edge[-] (t_1)
    (para) edge[-] (t_2) %
    (t0) edge[-] (t_1) %
    (t0) edge[-] (t_2) %
    (t1) edge[-,fore] (t_1) %
    (t1) edge[-] (t_2) %
    (s0) edge[-] (s) %
    (s1) edge[-] (s) %
    (s0) edge[-] (t0) %
    (s1) edge[-] (t1) %
    (s) edge[-] (para) %
    (i0) edge[-,gray,very thin] (para) %
    (i1) edge[-,gray,very thin] (para) %
  }\label{eq:para}
\end{equation}
to be thought of as a move from the bottom position $X$ \hfill
  \diagramme[stringdiag={.4}{.6}]{baseline=(s.south)}{}{ %
    \canal{s0} \& \& \& \joueur{s} \& \& \& \& \canal{s1} \& \& %
  }{%
    (s) edge[-] (s1) edge[-] (s0) %
  }
\linebreak\noindent to the top position $Y$
\hfill
  \diagramme[stringdiag={.4}{.6}]{}{
            \node[diagnode,at= (t1.south east)] {\ \ \ .} ; %
  }{%
    \& \& \&  \& \joueur{t_2} \\
    \canal{t0} \& \& \& \& \& \& \& \canal{t1} \& \& \\ 
    \& \& \joueur{t_1}  \&  \\
  }{%
    (t0) edge[-] (t_1) %
    (t0) edge[-] (t_2) %
    (t1) edge[-] (t_1) %
    (t1) edge[-] (t_2) %
  }

  The whole move may be viewed as a cospan $X \into P \otni Y$ in the
  category of diagrams (technically a presheaf category $\Chat$). Both
  legs of the cospan are actually monic arrows in $\Chat$, as will be
  the case for all cospans considered here. The vertical lines
  represent dots (channels and players) moving in time, upwards. So
  for example the left- and right-hand borders are just channels
  evolving in time, not noticing that the represented player forks
  into two. The surfaces spread between those vertical lines represent
  links (edges in the involved positions) evolving in time. For
  example, each link here divides into two when the player forks, thus
  representing the fact that both of the newly created players retains
  knowledge of the corresponding name. As for channels known to a
  player, the players and channels touching the black triangle are
  ordered: in particular there are different `ports' for the initial
  player and its two avatars.

There is of course an instance $\paran$ of forking for
each $n$, according to the number of channels known to the player.
Again, these explanations are very informal, but the diagrams have a
very precise combinatorial definition.

The above forking move has two \emph{basic} sub-moves, \emph{left} and
\emph{right half-forking}, respectively represented by the diagrams
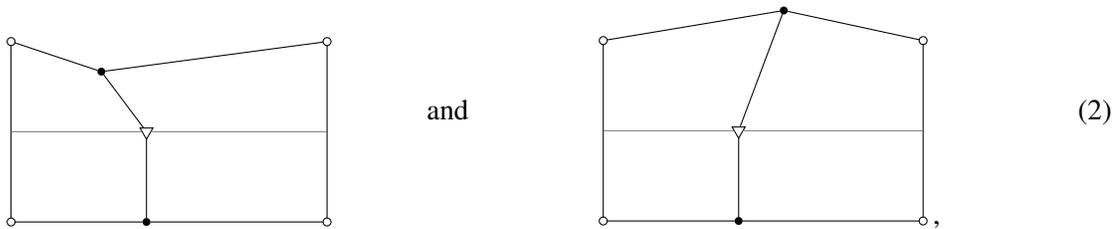
\begin{equation}
    \diagramme[stringdiag={.4}{.6}]{}{
  }{%
    \& \& \&  \&  \\
    \canal{t0} \& \& \& \& \& \& \& \canal{t1} \& \& \\ 
    \& \& \joueur{t_1}  \&  \\
    \& \ \& \\
    \coord{i0} \& \& \& \coupparacreux{para} \& \& \& \& \coord{i1} \& \& \\ 
    \& \ \& \\
    \& \&  \\
    \canal{s0} \& \& \& \joueur{s} \& \& \& \& \canal{s1} \& \& 
  }{%
    (t0) edge[-] (t_1) %
    (t1) edge[-] (t_1) %
    (s0) edge[-] (s) %
    (s1) edge[-] (s) %
    (s0) edge[-] (t0) %
    (s1) edge[-] (t1) %
    (s) edge[-] (para) %
    (para) edge[-] (t_1) %
    (i0) edge[-,gray,very thin] (para) %
    (i1) edge[-,gray,very thin] (para) %
  }
\mbox{and} \hspace*{.1\textwidth}
    \diagramme[stringdiag={.4}{.6}]{}{
    \node[diagnode,at= (s1.south east)] {\ \ \ ,} ; %
  }{%
    \& \& \&  \& \joueur{t_2} \\
    \canal{t0} \& \& \& \& \& \& \& \canal{t1} \& \& \\ 
    \& \&   \&  \\
    \& \ \& \\
    \coord{i0} \& \& \& \coupparacreux{para} \& \& \& \& \coord{i1} \& \& \\ 
    \& \ \& \\
    \& \&  \\
    \canal{s0} \& \& \& \joueur{s} \& \& \& \& \canal{s1} \& \& 
  }{%
    (t0) edge[-] (t_2) %
    (t1) edge[-] (t_2) %
    (s0) edge[-] (s) %
    (s1) edge[-] (s) %
    (s0) edge[-] (t0) %
    (s1) edge[-] (t1) %
    (s) edge[-] (para) %
    (para) edge[-] (t_2) %
    (i0) edge[-,gray,very thin] (para) %
    (i1) edge[-,gray,very thin] (para) %
  }\label{eq:paraviews}
\end{equation}
which represent what each of the `avatars' of the forking player sees
of the move. If a play contains both of the latter moves, then it
contains the full move~\eqref{eq:para}.  Forking, being the only move
with more than one player in the final position, is the only one
subject to such a decomposition. We call $\paraln$ and $\pararn$ the
respective instances of the left-hand and right-hand basic moves for a
player knowing $n$ names.

Let us now review the other basic moves of the game, which are also
full. As for forking, there is an instance of each of them following
the number $n$ of channels known to the player, and we only show the
case $n = 2$.  First, we have the \emph{tick} move $\tickn$, whose
role is to define successful plays, and the usual \emph{name
  creation}, or \emph{restriction} $\nu$ of CCS, here $\nun$. They are
graphically represented as
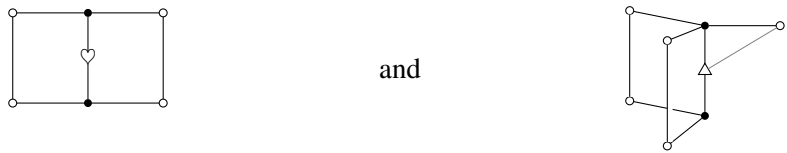
\begin{center}
  \diagramme[stringdiag={.6}{1}]{}{
    \path[-] (a) edge (a') %
    edge (p) %
    (tick) edge[shorten <=-1pt] (p) edge[shorten <=-1pt] (p') %
    (p') edge (a') edge (b') %
    (b) edge (p) edge (b') %
    ; %
  }{ %
    \canal{a'} \& \joueur{p'} \& \canal{b'} \\ %
    \& \couptick{tick} \& \\ %
    \canal{a} \& \joueur{p} \& \canal{b} \\
  }{%
  }
  \hfil and \hfil %
  \diagramme[stringdiag={.2}{.5}]{}{%
    \node[coordinate] (inter) at (intersection cs: %
    first line={(s)-- (s0)}, %
    second line={(s1)-- (t1)}) {} ; %
    \path[draw] (s) edge (inter) ; %
    \path[-,draw] %
    (s1) edge (s) %
    (t1) edge (t) %
    (t0) edge (t) %
    (t2) edge (t) %
    (t) edge (nu) %
    (s) edge (nu) %
    (s0) edge[fore] (t0) %
    (s1) edge[fore] (t1) %
    (nu) edge[gray,very thin] (t2) %
    (inter) edge (s0) %
    ; %
  }{%
    \canal{t0}     \& \&  \\
    \& \& \joueur{t} \& \& \canal{t2} \& \\
    \& \canal{t1} \&    \\
    \\
    \& \& \coupnu{nu} \& \&  \& \\
    \\
    \canal{s0}     \& \&  \& \&  \\
    \& \& \joueur{s} \& \& \& \\\\
    \& \canal{s1} \&  \& \&  
  }{%
  }%
\end{center}
respectively.  We finally have input and output, $a.P$ and $\abar.P$
in CCS, respectively $\iotaposni$ and $\iotanegni$ here ($n$ is the
number of known channels, $i$ is the index of the channel bearing the
synchronisation).  Here, output on the right-hand name and input on
the left-hand name respectively look like
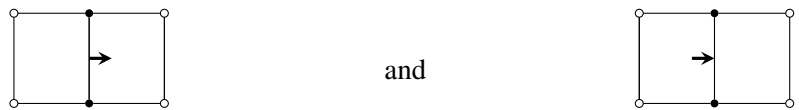
\begin{center}
  \diagramme[stringdiag={.6}{1}]{}{
    \path[-] (a) edge (a') %
    edge (p) %
    (in.west) edge (p) edge (p') %
    (p') edge (a') edge (b') %
    (b) edge (p) edge (b') %
    ; %
    \path (in) --  (bin) node[coordinate,pos=.3] (intip) {} ; %
    \path[-] (in) edge[->,>=stealth,very thick] (intip) ; %
    \foreach \x/\y in {p/p',b/b'} \path[-] (\x) edge (\y) ; %
  }{ %
    \canal{a'} \& \joueur{p'} \& \canal{b'} \\ %
    \& \coupout{in}{180} \& \coord{bin} \\ %
    \canal{a} \& \joueur{p} \& \canal{b} \\
  }{%
  }
  \hfil and \hfil
  \diagramme[stringdiag={.6}{1}]{}{
    \path[-] (a) edge (a') %
    edge (p) %
    (in.west) edge (p) edge (p') %
    (p') edge (a') edge (b') %
    (b) edge (p) edge (b') %
    ; %
    \node[diagnode,at= (b.south east)] {\ \ \ .} ; %
    \path (bin) --  (in) node[coordinate,pos=.7] (intip) {} ; %
    \path[-] (intip) edge[->,>=stealth,very thick] (in) ; %
    \foreach \x/\y in {p/p',a/a'} \path[-] (\x) edge (\y) ; %
  }{ %
    \canal{a'} \& \joueur{p'} \& \canal{b'} \\ %
    \coord{bin} \& \coupout{in}{0} \& \\ %
    \canal{a} \& \joueur{p} \& \canal{b} \\
  }{%
  }
\end{center}
We have now defined all basic and full moves, and move on to define
closed-world moves. Forking, name creation, and tick are all
closed-world moves, and there is only one more closed-world move,
which models CCS-like synchronisation.  For all $n$ and $m$,
representing the numbers of channels known to the players involved in
the synchronisation, and for all $i \in n$, $j \in m$ (seeing $n$ and
$m$ as finite ordinals), there is a \emph{synchronisation} $\taunimj$,
represented, (in the case where one player outputs on channel $1 \in
2$ and the other inputs on channel $0 \in 1$,) by
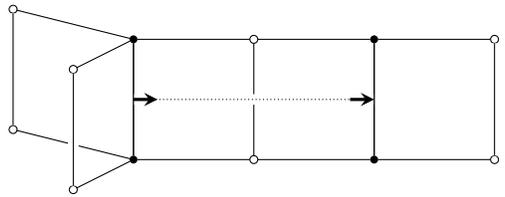
\begin{center}
  \diagramme[stringdiag={.4}{.8}]{}{%
    \node[coordinate] (inter) at (intersection cs: %
    first line={(s)-- (s0)}, %
    second line={(s1)-- (t1)}) {} ; %
    \path[draw] (s) edge (inter) ; %
      \path[-] %
      (s1) edge (s) %
      (s0) edge (inter) %
      (s2) edge (s) %
      (s2) edge (s') %
      (t1) edge (t) %
      (t0) edge (t) %
      (t2) edge (t) %
      (t2) edge (t') %
      (t) edge (iota.west) %
      (s) edge (iota.west) %
      (s') edge (iota'.east) %
      (t') edge (iota'.east) %
      (t'0) edge (t') %
      (s'0) edge (s') %
      (s'0) edge[fore] (t'0) %
      (s0) edge[fore] (t0) %
      (s1) edge[fore] (t1) %
      (s2) edge (t2) %
      ; %
    \path[-] %
    (iota) edge[fore,densely dotted] (iota') %
    ; %
    \node[diagnode,at= (s'0.south east)] {\ \ \ .} ; %
    \path (iota) --  (iota') node[coordinate,pos=.1] (iotatip) {} node[coordinate,pos=.9] (iotatip') {} ; %
      \path[-] (iota) edge[->,>=stealth,very thick] (iotatip) ; %
      \path[-] (iotatip') edge[->,>=stealth,very thick] (iota') ; %
      \foreach \x/\y in {s/t,s'/t'} \path[-] (\x) edge (\y) ; %
  }{%
    \canal{t0} \& \&  \&  \\
    \& \& \joueur{t} \& \& \canal{t2} \& \& %
    \joueur{t'} \& \&  \canal{t'0} \\
    \& \canal{t1} \&  \& \\
    \& \& \coupout{iota}{0} \& \&  \& \& \coupin{iota'}{0} \\
    \canal{s0} \& \&  \& \& \&  \& \\
    \& \& \joueur{s} \& \& \canal{s2} \& \& \joueur{s'} \& \& \canal{s'0} \\ \& \canal{s1} \& \& %
  }{%
  }%
  \end{center}
  Let us emphasise here that the dotted wire in the picture is
  actually a point in the formal representation (i.e., an element of
  the corresponding presheaf).

  Table~\ref{tab:moves} summarises the various classes of moves, and
  altogether they form the set of \emph{moves}.
  \begin{table}[t]
    \centering
    \begin{tabular}[t]{c|c|c}
      Basic & Full & Closed-world \\  \hline & &  \\[-.5em]
      \begin{minipage}[c]{0.2\linewidth} \centering
        Left half-forking \\
        Right half-forking 
      \end{minipage}
      & 
      Forking & Forking \\[1em]
      \begin{minipage}[c]{0.2\linewidth} \centering
        Input \\
        Output
      \end{minipage}
      & 
      \begin{minipage}[c]{0.2\linewidth} \centering
        Input \\
        Output
      \end{minipage}
      & 
      Synchronisation \\[1em]
      Tick & Tick & Tick \\
      Name creation & Name creation & Name creation 
    \end{tabular}
    \caption{Summary of classes of moves}
    \label{tab:moves}
  \end{table}

\subsection{Plays}
We now sketch how plays are defined as glueings of moves. We start
with the following example, depicted in Figure~\ref{fig:exex}. The
initial position consists of two players $p_1$ and $p_2$ sharing
knowledge of a name $a$, each of them knowing another name, resp.\
$a_1$ and $a_2$.  The play consists of four moves: first $p_1$ forks
into $p_{1,1}$ and $p_{1,2}$, then $p_2$ forks, and then $p_{1,1}$
does a left half-fork into $p_{1,1,1}$; finally $p_{1,1,1}$
synchronises (as the sender) with $p_{2,1}$. Now, we reach the limits
of the graphical representation, but the order in which the forks of
$p_1$ and $p_2$ occur is irrelevant: if $p_2$ had forked before $p_1$,
we would obtain the same play. This means that glueing the various
parts of the picture in Figure~\ref{fig:exex} in different orders
formally yields the same result (although there are subtle issues in
representing this result graphically in a canonical way).

Now, recall that moves may be seen as cospans $X \into M \otni
Y$. Now, consider an \emph{extended} notion of move, which may occur
in a larger position than just one player (two for synchronisation).
For example, the moves in Figure~\ref{fig:exex} are extended moves in
this sense.
We may now state:
  \begin{defn}
    A \emph{play} is an embedding $X_0 \into U$ in the category
    $\Chat$ of diagrams, isomorphic to a possibly denumerable
    `composition' of moves 
    in the (bi)category $\Cospan{\Chat}$ of cospans in $\Chat$, i.e.,
    obtained as a colimit:
      \begin{center}
        \diagramme[diagorigins={.6}{1.3}]{}{%
          \pbk{X_0}{U}{X_nnn} %
        }{%
          |(X_0)| X_0 \& \& |(X_1)| X_1 \& \ldots \& |(X_n)| X_n \&  \& |(X_nn)| X_{n+1} \& \& |(X_nnn)| X_{n+2} \& \ldots \\
          \& |(M_0)| M_0 \& \& \ldots \& \& |(M_n)| M_n \& \& |(M_nn)| M_{n+1} \& \ldots \\
          \& {\ } \\
          \& \& \& \& |(U)| U, %
        }{%
          (X_0) edge[into] (M_0) %
          (X_1) edge[linto] (M_0) %
          (X_n) edge[into] (M_n) %
          (X_nn) edge[linto] (M_n) %
          (X_nn) edge[into] (M_nn) %
          (X_nnn) edge[linto] (M_nn) %
          (M_0) edge[into] (U) %
          (M_n) edge[linto] (U) %
          (M_nn) edge[linto] (U) %
        }
      \end{center}
      where each $X_i \into M_i \otni X_{i+1}$ is 
      an extended move.
    \end{defn}
    Notation: we often denote plays just by $U$, leaving the embedding
    $X \into U$ implicit.
    \begin{rk}
      For finite plays, one might want to keep track not only of the
      initial position, but also of the final position. This indeed
      makes sense. Finite plays then compose `vertically', and form a
      double category. But infinite plays do not really have any final
      position, which explains our definition.
    \end{rk}
\begin{figure}[t]
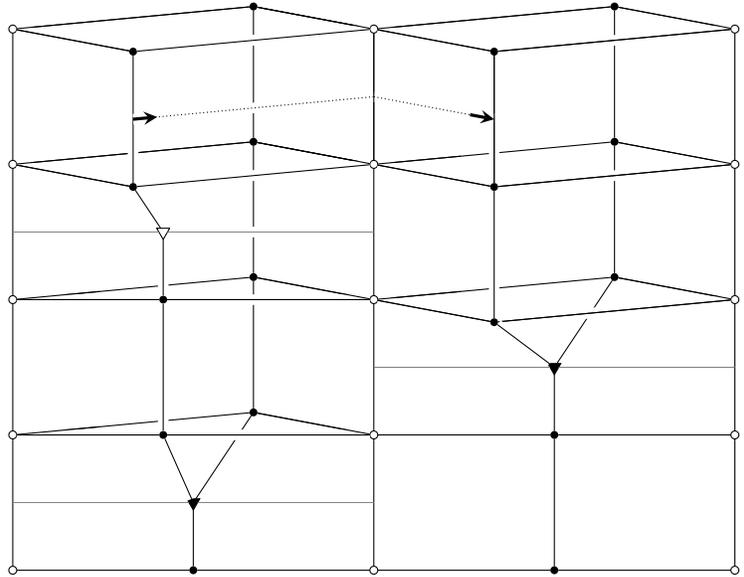

  \begin{center}
    \diagramme[stringdiag={.3}{.4}]{}{ \node[diagnode,at= (a2.south
      east)] {\ \ \ .} ; %
      \foreach \x/\y in {a1/a1',a1'/a1'',a1''/a1''',a1'''/a1'''',a2/a2',a2'/a2'',a2''/a2''',a2'''/a2'''',a/a',a'/a'',a''/a''',a'''/a'''',%
        p1/parap1,parap1/p12,p12/p12',p12'/p12'',a1/p1,p1/a,a1'/p12,p12/a',a1''/p12',p12'/a'',a/p2,p2/a2,a'/p2',%
        p2/p2',p2'/a2',p2'/parap2,parap2/p22,a''/p22,p22/a2'',a'''/p22',p22'/a2''',p12''/p12''',p22/p22',p22'/p22'',%
        a''''/p22'',p22''/a2''''%
      } 
      \path[-] (\x) edge (\y) ; %
      \foreach \a/\p/\b in {a1'''/p12''/a''',
        a1'''/p111/a''',a1''''/p12'''/a'''',
        a1''''/p111'/a'''',%
        a1'/p11/a',a1''/p11'/a'',a''/p21/a2'',a'''/p21'/a2''',a''''/p21''/a2''''} 
      \path[-]  (\p) edge[fore] (\a) edge[fore] (\b) %
      ; %
      \foreach \a/\p/\b in {a1'''/p12''/a''',
        a1'''/p111/a''',a1''''/p12'''/a'''',
        a1''''/p111'/a'''',%
        a1'/p11/a',a1''/p11'/a'',a''/p21/a2'',a'''/p21'/a2''',a''''/p21''/a2''''} 
      \path[-]  (\p) edge (\a) edge (\b) %
      ; %
      \foreach \x/\y in {parap1/p11,p11/p11',p11'/parap11,
        p111/a''',%
        p111'/a'''',parap2/p21,p21/p21',p21'/p21'',p111/p111'} %
      \path[-] (\x) edge[fore] (\y) %
      ; %
      \foreach \x/\y in {
        p12/a',p12'/a'',p12''/a''',p22/a2'',p22'/a2''',p22''/a2'''',
        p12'''/a'''',p22''/a''''} %
      \path[-] (\x) edge (\y) ; %
      \foreach \x/\y in {p22/a'',p22'/a''',p12/a1',p12'/a1''} %
      \path[-] (\x) edge[shorten <=2cm] (\y) ; %
      \foreach \para/\a/\ai in {parap1/a1parap1/aparap1,
        parap2/a2parap2/aparap2, parap11/a1parap11/aparap11}
      \path[draw] (\para) edge[-,gray,very thin] (\a) %
      edge[-,gray,very thin] (\ai) %
      ; %
      \path[-] (parap11) edge[-] (p111) ; %
      \path[-] (parap11) edge[-,fore,gray,very thin,shorten <=1pt, shorten >=1pt] (aparap11) ; %
      \path[-] (out11) edge[fore,densely dotted] (aout11) %
      (aout11) edge[fore,densely dotted] (in11) %
      ; %
      \path (out11) --  (aout11) node[coordinate,pos=.1] (out11tip) {} ; %
      \path (in11) -- (aout11) node[coordinate,pos=.2] (in11tip) {} ; %
      \path[-] (out11) edge[->,>=stealth,very thick] (out11tip) ; %
      \path[-] (in11tip) edge[->,>=stealth,very thick] (in11) ; %
      \foreach \x/\y in {a'''/a'''',p21'/p21''} \path[-] (\x) edge (\y) ; %
    }{%
      \& \& \& \& \& \& \& \& \joueur{p12'''} \& \& \& \& \& \& \& \& \& \& \& \& \joueur{p22''} \& \& \\
      \canal{a1''''} \& \& \& \& \& \& 
      \& \& \& \& \& \& \canal{a''''} \& \& \& \& \& \& \& \& \& \& \& \& \canal{a2''''} \\
      \& \& \& \& \joueur{p111'} \& \& \& \& \& \& \& \& \& \& \& \& \joueur{p21''} \& \& \\
      \& \& \\
      \coord{a1out11} \& \& \& \&  \&  \& \& \& \& \& \& \& \coord{aout11} \& \& \& \& \& \& \\
      \coord{a1out11} \& \& \& \& \coupout{out11}{4} \&  \& \& \& \& \& \& \&  \& \& \&  \& \coupin{in11}{-7} \& \& \\
      \& \& \& \& \& \& \& \& \joueur{p12''} \& \& \& \& \& \& \& \& \& \& \& \& \joueur{p22'} \& \& \\
      \canal{a1'''} \& \& \& \& \& \& 
      \& \& \& \& \& \& \canal{a'''} \& \& \& \& \& \& \& \& \& \& \& \& \canal{a2'''} \\
      \& \& \& \& \joueur{p111} \& \& \& \& \& \& \& \& \& \& \& \& \joueur{p21'} \& \& \\
      \& \& \\
      \coord{a1parap11} \& \& \& \& \& \coupparacreux{parap11} \& \& \& \& \& \& \& \coord{aparap11} \& \& \& \& \& \& \\
      \& \& \\
      \& \& \& \& \& \& \& \& \joueur{p12'} \& \& \& \& \& \& \& \& \& \& \& \& \joueur{p22} \& \& \\
      \canal{a1''} \& \& \& \& \& \joueur{p11'} \& \& \& \& \& \& \& \canal{a''} \& \& \& \& \& \& \& \& \& \& \& \& \canal{a2''} \\
      \& \& \& \&  \& \& \& \& \& \& \& \& \& \& \& \& \joueur{p21} \& \& \\
      \& \& \\
      \coord{a1parap2} \& \& \& \& \&  \& \& \& \& \& \& \& \coord{aparap2} \& \& \& \& \& \& \couppara{parap2} \& \& \& \& \& \& \coord{a2parap2} \\
      \& \& \\
      \& \& \& \& \& \& \& \& \joueur{p12} \& \& \& \& \& \& \& \& \& \& \& \&  \& \& \\
      \canal{a1'} \& \& \& \& \& \joueur{p11} \& \& \& \& \& \& \& \canal{a'} \& \& \& \& \& \& \joueur{p2'} \& \& \& \& \& \& \canal{a2'} \\
      \& \& \& \&  \& \& \& \& \& \& \& \& \& \& \& \&  \& \& \\
      \& \& \\
      \coord{a1parap1} \& \& \& \& \&  \& \couppara{parap1} \& \& \& \& \& \& \coord{aparap1} \& \& \& \& \& \&  \& \& \& \& \& \&  \\
      \& \& \\
      \& \& \& \& \& \& \& \&  \& \& \& \& \& \& \& \& \& \& \& \&  \& \& \\
      \canal{a1} \& \& \& \& \&  \& \joueur{p1} \& \& \& \& \& \& \canal{a} \& \& \& \& \& \& \joueur{p2} \& \& \& \& \& \& \canal{a2} \\
    }{%
    }
  \end{center}
  \caption{An example play}
\label{fig:exex}
\end{figure}
  
\begin{wrapfigure}{r}{0pt}
      \diag(.6,.8){%
        |(U)| U \& |(V)| V \\
        |(X)| X \& |(Y)| Y. %
      }{%
        (X) edge[into] (U) edge[labelu={h}] (Y) %
        (U) edge[labelu={k}] (V) %
        (Y) edge[into] (V) %
      }
\end{wrapfigure}
Let a morphism $(X \into U) \to (Y \into V)$ of plays be a pair
$(h,k)$ making the diagram on the right commute in $\Chat$.  This
permits both inclusion `in width' and `in height'. E.g., the play
consisting of the left-hand basic move in~\eqref{eq:paraviews} embeds
in exactly two ways into the play of Figure~\ref{fig:exex}. (Only two
because the image of the base position must lie in the base position
of the codomain.)  We have:
\begin{prop}
  Plays and morphisms between them form a category $\E$.
\end{prop}
There is a projection functor $\E \to \B$ mapping each play $X
\into U$ to its base position $X$.  This functor has a section, which
is an embedding $\B \into \E$, mapping each position $X$ to the empty
play $X \into X$ on $X$.

\begin{rk}[Size]
  The category $\CE$ is only locally small. Since presheaves on a
  locally small category are less well-behaved than on a small
  category, we will actually consider a skeleton of $\CE$.
  Because $\CE$ consists only of denumerable presheaves, this skeleton
  is a small category. Thus, our presheaves in the next section may be
  understood as taken on a small category.
\end{rk}

\begin{rk}
  Plays are not very far from being just (infinite) abstract syntax
  trees (or forests) `glued together along channels'.
\end{rk}

\subsection{Relativisation}
If we now want to restrict to plays over a given base position $X$, we may consider
\begin{defn}
  Let the category $\CEX$ have
  \begin{itemize}
  \item as objects pairs of a play $Y \into U$ and a morphism $Y \to X$,
  \item as morphisms $(Y \into U) \to (Y' \into U')$ all pairs $(h,k)$
    making the following diagram commute:
    \begin{center}
  \diag(.25,1.5){%
    |(U)| U \& \& |(U')| U' \\
    \\
    |(Y)| Y \& \& |(Y')| Y' \\
    \& |(X)| X %
  }{%
    (U) edge[labelu={k}] (U') %
    (Y) edge[into] (U) %
    edge [labelu={h}] (Y') %
    edge (X)
    (Y') edge[into] (U') %
    edge (X)
  }%
    \end{center}
    in $\Chat$.
  \end{itemize}
\end{defn}
We will usually abbreviate $U \otni Y \into X$ as just $U$ when no
ambiguity arises.  As for morphisms of positions, in full generality,
$h$ and $k$, as well as the morphisms $Y \to X$, do not have to be
injective. However, for the purpose of this expository paper, let us
again restrict to injective $h, k$, and $Y \to X$.
\begin{ex}
  Let $X$ be the position
  \diagramme[stringdiag={.1}{.2}]{baseline=(s_2.south)}{}{ %
    \canal{s0} \& \& \& \joueur{s_1} \& \& \& \& %
    \canal{s1} \& \& \& \joueur{s_2} \& \& \& \& %
    \canal{s2} \& \& \& \joueur{s_3} \& \& \& \& %
    \canal{s3} %
  }{%
    (s_1) edge[-] (s1) edge[-] (s0) %
    (s_2) edge[-] (s1) edge[-] (s2) %
    (s_3) edge[-] (s2) edge[-] (s3) %
  }.  %
  The play in Figure~\ref{fig:exex}, say $Y \into U$, equipped with
  the injection $Y \into X$ mapping the two players of $Y$ to the two
  leftmost players of $X$, is an object of $\CEX$.
\end{ex}
\section{Innocent strategies as sheaves}
Now that the category of plays is defined, we move on to defining
innocent strategies.  There is a notion of a Grothendieck
\emph{site}~\cite{MM}, which consists of a category equipped with a
(generalised) topology. On such sites, one may define a category of
sheaves, which are very roughly the presheaves that are continuous for
the topology. We claim that there is a topology on each $\CEX$, for
which sheaves adequately model innocent strategies. Fortunately, in
our setting, sheaves admit a simple description, so that in this
expository paper we can avoid the whole machinery. 

\subsection{Innocent strategies}
\begin{defn}
  A \emph{view} is a finite `composition' $n \into V$ of basic moves
  in $\Cospan{\Chat}$.
\end{defn}
  \begin{ex}\label{ex:covpara}
    Forking~\eqref{eq:para} has two non-trivial views, namely the
    basic moves~\eqref{eq:paraviews}.
  \end{ex}
\begin{ex}
  In Figure~\ref{fig:exex}, the left-hand branch contains a view
  consisting of three basic moves: two $\paralof{2}$ and an input.
\end{ex}

  \begin{ex}\label{ex:cov}
    The embeddings
    
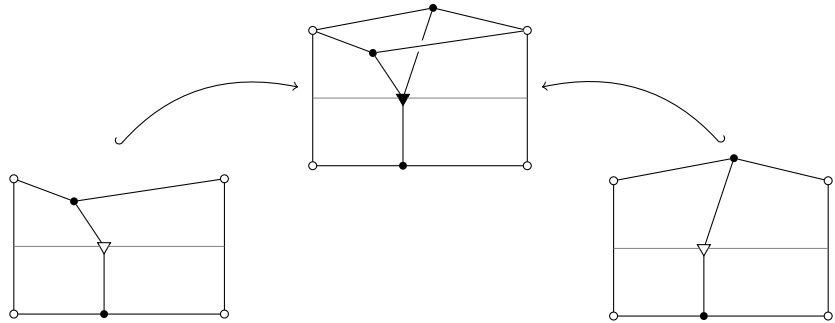
\begin{center}
    \begin{tikzpicture}
      \matrix(m)[stringdiag={.3}{.4}]{
        \& \& \&  \& \joueur{t_2} \\
        \canal{t0} \& \& \& \& \& \& \& \canal{t1}  \\ 
        \& \& \joueur{t_1}  \&  \\
        \& \ \& \\
        \coord{i0} \& \& \& \couppara{para} \& \& \& \& \coord{i1} \\ 
        \& \ \& \\
        \& \&  \\
        \canal{s0} \& \& \& \joueur{s} \& \& \& \& \canal{s1}  \\
      } ;
      \path[->] %
      (para) edge[-] (t_1) %
      (para) edge[-] (t_2) %
      (t0) edge[-] (t_1) %
      (t0) edge[-] (t_2) %
      (t1) edge[-,fore] (t_1) %
      (t1) edge[-] (t_2) %
      (s0) edge[-] (s) %
      (s1) edge[-] (s) %
      (s0) edge[-] (t0) %
      (s1) edge[-] (t1) %
      (s) edge[-] (para) %
      (i0) edge[-,gray,very thin] (para) %
      (i1) edge[-,gray,very thin] (para) %
      ; %
      \matrix(m1)[stringdiag={.3}{.4}] at (-4,-2) {
        \& \& \&   \\
    \canal{t0i} \& \& \& \& \& \& \& \canal{t1i} \\ 
    \& \& \joueur{t_1i}  \\
    \& \ \& \\
    \coord{i0i} \& \& \& \coupparacreux{parai} \& \& \& \& \coord{i1i} \\
    \& \ \& \\
    \& \&  \\
    \canal{s0i} \& \& \& \joueur{si} \& \& \& \& \canal{s1i} \\
  } ; %
  \path[->] %
    (t0i) edge[-] (t_1i) %
    (t1i) edge[-] (t_1i) %
    (s0i) edge[-] (si) %
    (s1i) edge[-] (si) %
    (s0i) edge[-] (t0i) %
    (s1i) edge[-] (t1i) %
    (si) edge[-] (parai) %
    (parai) edge[-] (t_1i) %
    (i0i) edge[-,gray,very thin] (parai) %
    (i1i) edge[-,gray,very thin] (parai) %
    ; %
    \matrix(m2)[stringdiag={.3}{.4}] at (4,-2) {%
    \& \& \&  \& \joueur{t_2ii} \\
    \canal{t0ii} \& \& \& \& \& \& \& \canal{t1ii}  \\ 
    \& \&   \&  \\
    \& \ \& \\
    \coord{i0ii} \& \& \& \coupparacreux{paraii} \& \& \& \& \coord{i1ii} \\ 
    \& \ \& \\
    \& \&  \\
    \canal{s0ii} \& \& \& \joueur{sii} \& \& \& \& \canal{s1ii} \\
  } ; %
  \path[->] %
    (t0ii) edge[-] (t_2ii) %
    (t1ii) edge[-] (t_2ii) %
    (s0ii) edge[-] (sii) %
    (s1ii) edge[-] (sii) %
    (s0ii) edge[-] (t0ii) %
    (s1ii) edge[-] (t1ii) %
    (sii) edge[-] (paraii) %
    (paraii) edge[-] (t_2ii) %
    (i0ii) edge[-,gray,very thin] (paraii) %
    (i1ii) edge[-,gray,very thin] (paraii) %
  ; %
  \path[->] (m1.north) edge[into,bend left] (m.west) %
  (m2.north) edge[linto,bend right] (m.east) %
  ; %
  \end{tikzpicture}
    \end{center}
      have views as domains.
  \end{ex}

For any position $X$, let $\CVX$ be the subcategory of $\CEX$
consisting of views. 
\begin{defn}
  Let the category $\SS_X$ of \emph{innocent strategies} on $X$ be the
  category $\CVXhat$ of presheaves on $\CVX$.
\end{defn}
A possible interpretation is that for a presheaf $F \in \CVXhat$ and
view $V \in \CV_X$, $F (V)$ is the set of possible \emph{states} of
the strategy $F$ after playing $V$.

It might thus seem that we could content ourselves with defining only
views, as opposed to full plays.  However, in order to define
interactive equivalences in Section~\ref{sec:inter}, we need to view
innocent strategies as (particular) presheaves on the whole of $\CEX$.

\begin{wrapfigure}{r}{0pt}
  \Diag(1.1,1){%
  }{%
    |(C)| \C \& \& \& |(D)| \D \\
    \& |(E)| \E \& %
  }{%
    (C) edge[labelu={F}] (D) %
    edge[bend right,labelbl={G}] node[pos=.4] (G) {} (E) %
    (D) edge[bend right] node[pos=.5] (al) {}
    node[pos=.2,anchor=north] {$\scriptstyle H$} (E) %
    edge[bend left,labelbr={K}] node[pos=.5] (br) {} (E) %
    (br) edge[cell={0},labelr={\scriptstyle \alpha'}] (al) %
    (al) edge[cell={0},bend right=10,labelu={\scriptstyle
      \varepsilon}] (G) (br) edge[cell={.1}] node[pos=.6,anchor=south]
    {$\scriptstyle \alpha$} (G) %
  }
\end{wrapfigure}
The connection is as follows. Recall from MacLane~\cite{MacLane:cwm}
the notion of \emph{right Kan extension}. Given functors $F$ and $G$
as on the right, a right Kan extension $\Ran_F (G)$ of $G$ along $F$
is a functor $H \colon \D \to \E$, equipped with a natural
transformation $\epsilon \colon HF \to G$, such that for all functors
$K \colon \D \to \E$ and transformations $\alpha \colon KF \to G$,
there is a unique $\alpha' \colon K \to H$ such that $\alpha =
\epsilon \mathbin{\rond_1} (\alpha' \rond \id_F)$, where $\rond_1$ is
vertical composition of natural transformations.  Now, precomposition
with $F$ induces a functor $\Cat(F,\E) \colon \Cat (\D,\E) \to \Cat
(\C,\E)$, where $\Cat (\D,\E)$ is the category of functors $\D \to \E$
and natural transformations between them. When $\E$ is complete, right
Kan extensions always exist (and an explicit formula for our setting
is given below), and choosing one of them for each functor $\C \to \E$
induces a right adjoint to $\Cat (F,\E)$. Furthermore, it is known
that when $F$ is full and faithful, then $\epsilon$ is a natural
isomorphism, i.e., $HF \iso G$, which entails that $\Ran_F$ is a full
essential embedding.

Returning to views and plays, the embedding $i_X \colon \CVX \into
\CEX$ is full, so right Kan extension along $\op{i_X} \colon \op{\CVX}
\to \op{\CEX}$ induces a full essential embedding $\Ran_{\op{i_X}}
\colon \CVXhat \to \CEXhat$.  The (co)restriction of this essential
embedding to its essential image thus yields an essentially
surjective, fully faithful functor, i.e., an equivalence of
categories:
\begin{prop}\label{prop:synsem}
  The category $\SS_X$ is equivalent to the essential image of
  $\Ran_{\op{i_X}}$.
\end{prop}
The standard characterisation of right Kan extensions as
ends~\cite{MacLane:cwm} yields, for any $F \in \CVXhat$ and $U \in
\CEX$:
$$\Ran_{\op {i_X}}(F)(U) = \int_{V \in \CVX} F(V)^{\CEX(V,U)},$$
i.e., an element of $\Ran_{\op {i_X}}(F)$ on a play $U$ consists, for
each view $V$ and morphism $V \to U$, of an element of $F(V)$,
satisfying some compatibility conditions. In Example~\ref{ex:ran}
below, we compute an example right Kan extension.

The interpretation of strategies in terms of states extends: for any
presheaf $F \in \CEXhat$ and play $U \in \CE_X$, $F (U)$ is the set of
possible \emph{states} of the strategy $F$ after playing $U$.  That
$F$ is in the image of $\Ran_{\op{i_X}}$ amounts to $F (U)$ being a
compatible tuple of states of $F$ after playing each view of $U$.

\begin{ex}\label{ex:ran}
  Here is an example of a presheaf $F \in \CEXhat$ which is not
  innocent, i.e., not in the image of $\Ran_{\op{i_X}}$. Consider the
  position $X$ consisting of three players, say $x,y,z$, sharing a
  name, say $a$. Let $X_x$ be the subposition with only $x$ and $a$,
  and similarly for $X_y$, $X_z$, $X_{x,y}$, and $X_{x,z}$.  Let $I_x
  = (\iotaneg{1,0} \otni X_x \into X)$ be the play where $x$ inputs on
  $a$, and similarly let $O_y$ and $O_z$ be the plays where $y$ and
  $z$ output on $a$, respectively. Let $F (I_x) = F (O_y) = F (O_z) =
  1$ be singletons. Let now $S_{x,y} = (\tauof{1}{0}{1}{0} \otni
  X_{x,y} \into X)$ be the play where $x$ and $y$ synchronise on $a$
  ($x$ inputs and $y$ outputs), and similarly let $S_{x,z}$ be the
  play where $x$ and $z$ synchronise on $a$.  Let $F (S_{x,y}) = 2$ be
  a two-element set, and let $F (S_{x,z}) = \emptyset$.  Finally, let
  $F$ map any subplay of the above plays to a singleton, and any
  strict superplay to the empty set.

  This $F$ fails to be innocent on two counts.  First, since $x$ and
  $y$ accept to input and output in only one way, it is non-innocent
  to accept that they synchronise in more than one way. Formally,
  $S_{x,y}$ has two non-trivial views, $I_x$ and $O_y$, so since $F$
  maps the empty view to a singleton, $F (S_{x,y})$ should be
  isomorphic to $F (I_x) \times F (O_y) = 1 \times 1 = 1$. Second,
  since $x$ and $z$ accept to input and output, it is non-innocent to
  not accept that they synchronise. Formally, $F (S_{x,z})$ should
  also be a singleton. This altogether models the fact that in CCS,
  processes do not get to choose with which other processes they
  synchronise.

  The restriction of $F$ to $\CVX$, i.e., $F' = F \rond \op{i_X}$, in
  turn has a right Kan extension $F''$, which is innocent. (In
  passing, the unit of the adjunction $\Cat (\op{i_X},\Set) \dashv
  \Ran_{\op{i_X}}$ is a natural transformation $F \to F''$.) To
  conclude this example, let us compute $F''$. First, $F'$ only
  retains from $F$ its values on views. So, if $X_x$ denotes the empty
  view on $X_x$, $F' (X_x) = 1$, and similarly $F' (X_y) = F' (X_z) =
  1$. Furthermore, $F' (I_x) = F' (O_y) = F' (O_z) = 1$. Finally, for
  any view $V$ not isomorphic to any of the previous ones, $F' (V) =
  \emptyset$. So, recall that $F''$ maps any play $U \otni Y \into X$
  to $\int_{V \in \CVX} F' (V)^{\CEX(V,U)}$.  So, e.g., since the
  views of $S_{x,y}$ are subviews of $I_x$ and $O_y$, we have $F''
  (S_{x,y}) = F' (I_x) \times F' (O_y) = 1$. Similarly, $F'' (S_{x,z})
  = 1$.  But also, for any play $U$ such that all views $V \to U$ are
  subviews of either of $I_x$, $O_y$, or $O_z$, we have $F'' (U) =
  1$. Finally, for any play $U$ such that there exists a view $V \to
  U$ which is not a subview of any of $I_x$, $O_y$, or $O_z$, we have
  $F'' (U) = \emptyset$.  
\end{ex}
One way to understand Proposition~\ref{prop:synsem} is to view
$\CVXhat$ as the syntax for innocent strategies: presheaves on views
are (almost) infinite terms in a certain syntax (see
Section~\ref{subsec:lang} below).
On the other hand, seeing them as presheaves on plays will allow us to
consider their global behaviour: see Section~\ref{sec:inter} when we
restrict to the closed-world game.
Thus, right Kan extension followed by restriction to closed-world will
associate a semantics to innocent strategies.

So, we have defined for each $X$ the category $\SS_X$ of innocent
strategies on $X$. This assignment is actually functorial $\op\B \to
\CAT$, as follows (where $\CAT$ is the large category of locally small
categories). Any morphism $f \colon Y \to X$ induces a functor $f_!
\colon \CVY \to \CVX$ sending $(V \otni Z \to Y)$ to $(V \otni Z \to Y
\to X)$. Precomposition with $\op{(f_!)}$ thus induces a functor
$\SS_f \colon \CVXhat \to \CVYhat$.
\begin{prop}
  This defines a functor $\SS \colon \op\B \to \CAT$.
\end{prop}

But there is more: for any position, giving a strategy for each player
in it easily yields a strategy on the whole position. We call this
\emph{amalgamation} of strategies. Formally, consider any subpositions
$X_1$ and $X_2$ of a given position $X$, inducing a partition of the
players of $X$, i.e., such that $X_1 \cup X_2$ contains all the
players of $X$, and $X_1 \cap X_2$ contains none.  Then 
$\CVX$ is isomorphic to the coproduct $\CV_{X_1} +
\CV_{X_2}$. (Indeed, a view contains in particular an initial player
in $X$, which forces it to belong either in $\CV_{X_1}$ or in
$\CV_{X_2}$.) 
\begin{defn}
  Given strategies $F_1$ on $X_1$ and $F_2$ on $X_2$, let their
  \emph{amalgamation} be their image in $\CVX$ via the above
  equivalence, i.e., the copairing $$[F_1,F_2] : \op{\CVX} \iso
  (\CV_{X_1} + \CV_{X_2})^{op} \iso \op{\CV_{X_1}} + \op{\CV_{X_2}}
  \to \Set.$$
\end{defn}
By universal property of coproduct:
\begin{prop}
  Amalgamation yields an equivalence of categories $\CVXhat \equi
  \widehat{\CV_{X_1}} \times \widehat{\CV_{X_2}}.$
\end{prop}

\begin{ex}
  Consider again the position $X$ from Example~\ref{ex:ran}, and let
  $X_{y,z}$ be the subposition with only $y$ and $z$.  We have $\CVX
  \equi (\CV_{X_x} + \CV_{X_{y,z}})$, which we may explain by hand as
  follows.  A view on $X$ has a base player, $x$, $y$, or $z$, and so
  belongs either in $\CV_{X_x}$ or in $\CV_{X_{y,z}}$. Furthermore, if
  $V$ is a view on $x$ and $W$ is a view on $y$, then $\CVX (V,W) =
  \emptyset$ (and similarly for any pair of distinct players in $X$).

  Now, recall $F'$, the restriction of $F$ to $\CVX$. We may define
  $F_x \colon \op{\CV_{X_x}} \to \Set$ to be the restriction of $F'$
  along the (opposite of the) embedding $\CV_{X_x} \into \CVX$, and
  similarly $F_{y,z}$ to be the restriction of $F'$ along
  $\CV_{X_{y,z}} \into \CVX$. Observe that $F'$ sends any view $V$ on
  $x$ to $F_x (V)$, and similarly for views on $y$ and $z$, we
  conclude that $F'$ is actually the copairing $[F_x, F_{y,z}]$.
\end{ex}

Analogous reasoning leads to the following. For any $X$, let $\Pl (X)$
denote the set of pairs $(n,x)$, where $x$ is a player in $X$, knowing
$n$ names.  This yields the \emph{spatial decomposition} theorem,
where $n$ is abusively seen as the position with one player knowing
$n$ names:
\begin{thm}
  We have $\CVXhat \equi \prod_{(n,x) \in \Pl (X)} \CVnhat$.
\end{thm}

There is actually more structure than that, namely the functor $\SS$
is a stack~\cite{Vistoli}, but we do not need to spell out the
definition here.

\subsection{Temporal decomposition and languages}\label{subsec:lang}
Let us now describe \emph{temporal} decomposition.  The main goal
here is to sketch the interpretation of CCS in innocent
strategies. This material is not needed for Section~\ref{sec:inter}.

Recall that full moves are forking~\eqref{eq:para}, tick, name
creation, input, and output.
  \begin{defn}
    Let $\MMM_n$ be the set of all full moves starting from $n$.  For
    each $M \in \MMM_n$, let $\cod (M)$ be the final position of the
    corresponding move.
  \end{defn}
   Strictly speaking, $\MMM_n$ is a proper class, but we may easily
   choose one representative of each isomorphism class. For instance,
   all moves are actually representable presheaves in $\Chat$, so we
   may just pick these.

   To state the temporal decomposition theorem, we need a
   standard~\cite{Jacobs} categorical construction, the category of
   families on a given category $\C$. First, given a set $X$, consider
   the category $\Fam (X)$ with as objects $X$-indexed families of
   sets, i.e., sets $(Y_x)_{x \in X}$, and as morphisms $Y \to Z$
   families $(f_x \colon Y_x \to Z_x)_{x \in X}$ of maps. This
   category is equivalently described as the slice category $\Set /
   X$. To see the correspondence, consider any family $(Y_x)_{x \in
     X}$, and map it to the projection function $\sum_{x \in X} Y_x
   \to X$ sending $(x, y)$ to $x$. Generalising from sets $X$ to
   arbitrary categories $\C$, $\Fam (\C)$, has as objects families $f
   \colon Y \to \ob(\C)$ indexed by the objects of $\C$, and as
   morphisms $(Y, p) \to (Z, q)$ pairs of $u \colon Y \to Z$ and $v
   \colon Y \to \mor(\C)$, such that $q u = p$, $\dom (v (x)) = p
   (x)$, and $\cod (v (x)) = q (u (x)).$ Thus, any element $y \in Y$
   above $C \in \C$ is mapped to some $u (y) \in Z$ above $C' \in \C$,
   and this mapping is labelled by a morphism $v (y) \colon C \to C'$
   in $\C$.  We may now state:
  \begin{thm}\label{thm:temp}
    There is an equivalence of categories $\SS_n \equi \Fam
    \left(\prod_{M \in \MMM_n}\SS_{\cod (M)}\right ).$
 \end{thm}
 The main intuition for the proof is that a strategy is determined up
 to isomorphism by (i) its initial states, and (ii) what remains of
 them after each possible full move.  The family construction is what
 permits non-deterministic strategies: a given move may lead to
 different states.

 \begin{rk}
   The theorem almost makes strategies into a \emph{sketch} (on the
   category with positions as objects, finite compositions of extended
   moves as morphisms, and the $\MMM_X$'s as distinguished
   cones). Briefly, being a sketch would require a bijection of sets
   $\SS_n \iso \prod_{M \in \MMM_n}\SS_{\cod (M)}.$ Here, the
   bijection becomes an equivalence of categories, and the family
   construction sneaks in.
 \end{rk}

Putting the decomposition theorems together, we obtain
$$\SS_n \equi \Fam \left( \prod_{M \in \MMM_n} \prod_{(n', x') \in \Pl (\cod (M))} \SS_{n'}\right),$$
for all $n$. Considering a variant of this formula as a system of
equations will lead to our interpretation of CCS.  The first step is
to consider the formula as an endofunctor $F_0$ on $\Cat / \omega$,
where $\omega$ is the set of finite ordinals, seen as a discrete
category. This functor is defined on any family of categories $\X =
(\X_i)_{i \in \omega}$ by:
$$(F_0 (\X))_n = \Fam \left( \prod_{M \in \MMM_n} \prod_{(n', x') \in \Pl (\cod (M))} \X_{n'}\right).$$
Then, using the notation just before Theorem~\ref{thm:temp}, we
restrict attention to families $Y \to \ob (\C)$ where $Y$ is a finite
ordinal $n \in \omega$. Simultaneously, we restrict attention to
discrete categories $\X$, i.e., we see the formula as a endofunctor on
$\Set / \omega$, i.e., $\omega$-indexed families of sets. This yields, for any family $X = (X_i)_{i \in \omega}$,
$$(F (X))_n = \sum_{I \in \omega} \left( \prod_{M \in \MMM_n} \prod_{(n', x') \in \Pl (\cod (M))} X_{n'}\right)^I.$$

This endofunctor is polynomial~\cite{Kock01012011} and we now give a
characterisation of its final coalgebra.  Let for any category $\C$
the category $\OPsh{\C}$ be the functor category $\op\C \to \FinOrd$,
where $\FinOrd$ is the category of finite ordinals and monotone
functions between them. By composition with the embedding $\FinOrd
\into \Set$, we have an embedding $\OPsh {\C} \into \hat\C$. We have:
\begin{thm}
  The family $\ob(\OPsh {\CV_n})$ formed for each $n$ by the objects of
  $\OPsh {\CV_n}$ is a terminal coalgebra for $F$.
\end{thm}
By Lambek's lemma~\cite{Lambeklemma}, there is a bijection (between
the objects)
\begin{equation}
\OPsh{\CV_n} \iso \sum_{I \in \omega} \left (\prod_{M \in \MMM_n} \prod_{(n',x) \in \Pl (\cod (M))} \OPsh{\CV_{n'}} \right )^I.
\label{eq:syntax}
\end{equation}

In particular, the family $\OPsh{\CV_n}$ supports the operations of the grammar
\begin{mathpar}
\inferrule{\ldots \ \ \   n \vdash F_i \ \ \  \ldots \\ (\forall i \in I)}{
n \vdash \sum_{i \in I} F_i}~(I \in \omega) \and
\inferrule{\ldots \ \ n' \vdash F_{M,n',x} \ \ \ldots \\ {(\forall M \in \MMM_n, (n',x) \in \Pl (\cod (M)))}%
}{
n \vdash \langle(M,(n',x)) \mapsto F_{M,(n',x)}\rangle 
}~\cdot
\end{mathpar}
Here, $n \vdash F$ denotes a presheaf of finite ordinals on $\CV_n$.
The interpretation is as follows: given presheaves $F_0, \ldots,
F_{I-1}$, for $I \in \omega$, the leftmost rule constructs the finite
coproduct $\sum_{i \in I} F_i$ of presheaves (finite coproducts exist in
$\OPsh{\CV_n}$ because they do in $\FinOrd$). In particular, when $I$
is the empty ordinal, we sum over an empty set, so the rule
degenerates to
\begin{mathpar}
  \inferrule{ }{n \vdash \emptyset}~\cdot
\end{mathpar}
In terms of presheaves, this is just the constantly empty presheaf.

For the second rule, if for all $M, n', x$, we are given $F_{M,(n',x)}
\in \OPsh{\CV_{n'}}$, then $\langle (M,(n',x)) \mapsto
F_{M,(n',x)}\rangle$ denotes the image under~\eqref{eq:syntax} of
$$(1, 0 \mapsto M \mapsto (n', x) \mapsto F_{M,n',x}).$$
Here, we provide an element of the right-hand side
of~\eqref{eq:syntax}, consisting of the finite ordinal $I = 1 =
\ens{0}$, and the function mapping $(M, n', x)$ to $F_{M,n',x} \in
\OPsh{\CV_{n'}}$ (up to currying). That was for parsing; the
intuition is that we construct a presheaf with one initial state, $0$,
which maps any view starting with $(M, n', x)$, say $M; V$, to $F_{M,
  n', x} (V)$. Thus the $F_{M,n',x}$'s specify what remains of our
presheaf after each possible basic move. In particular, when all the
$F_{M,n',x}$'s are empty, we obtain a presheaf which has an initial
state, but which does nothing beyond it. We abbreviate it as $0 =
\langle \_ \mapsto \emptyset \rangle$.

\subsection{Translating CCS}
It is rather easy to translate CCS into this language.  First, define
CCS syntax by the natural deduction rules in Figure~\ref{fig:ccs},
where $\Names$ and $\Vars$ are two fixed, disjoint, and infinite sets
of \emph{names} and \emph{variables}; $\Xi$ ranges over finite
sequences of pairs $(x \colon n)$ of a variable $x$ and its
\emph{arity} $n \in \omega$; $\Gam$ ranges over finite sequences of
names; there are two judgements: $\Gam \vdash P$ for \emph{global}
processes, $\Xi;\Gam \vdash P$ for \emph{open} processes.  Rule
\textsc{Global} is the only rule for forming global processes, and
there $\Xi = (x_1 \colon \card{\Del_1}, \ldots, x_n \colon
\card{\Del_n})$. Finally, $\alpha$ denotes $a$ or $\abar$, for $a \in
\Names$, and $\floor{a} = \floor{\abar} = a$.
\begin{figure}[t]
  \centering
  \begin{mathpar}
\inferrule[\textsc{CCSApp}]{ }{\Xi ; \Gam \vdash x (a_1, \ldots, a_{n})}~((x \colon n) \in \Xi \mbox{\ and\ } a_1, \ldots, a_{n} \in \Gam)
  \and
\inferrule{\Xi; \Gam, a \vdash P 
}{ %
  \Xi;\Gam \vdash \nu a.P %
}
\and
\inferrule{%
  \ldots \\ \Xi; \Gam \vdash P_i  \\ \ldots \\ (\forall i \in I)
}{ %
  \Xi;\Gam \vdash \sum_{i \in I} \alpha_i.P_i %
}~(I \in \omega \mbox{\ and\ } \forall i \in I, \floor{\alpha_i} \in \Gam)
\and
\inferrule{\Xi; \Gam \vdash P \\
\Xi; \Gam \vdash Q 
}{ %
  \Xi;\Gam \vdash P|Q %
}
\and
\inferrule[\textsc{Global}]{\Xi ; \Del_1 \vdash P_1 \\ \ldots \\ \Xi ; 
  \Del_n \vdash P_n \\
  \Xi ; \Gam \vdash P 
}{
  \Gam \vdash \recin{x_1 (\Del_1) \eq P_1, \ldots,
  x_n (\Del_n) \eq P_n}{P} %
}
\end{mathpar}
  \caption{CCS syntax}
  \label{fig:ccs}
\end{figure}

First, we define the following (approximation of a) translation on
open processes, mapping each open process $\Xi; \Gam \vdash P$ to
$\transl{P} \in \OPsh{\CV_n}$, for $n = \card{\Gam}$.  This
translation ignores the recursive definitions, and we will refine it
below to take them into account. We proceed by induction on $P$,
leaving contexts $\Xi; \Gam$ implicit:
 $$\begin{array}[t]{r@{\ \mapsto\ }lll}
 x (a_1, \ldots, a_{k}) & \emptyset \\
   P|Q & \langle \begin{array}[t]{r@{\ \mapsto\ }l}
    (\paraof{n}, n, t_1) & \transl{P}, \\
    (\paraof{n}, n, t_2) & \transl{Q}, \\
    (\_, \_) & \emptyset \ \ \rangle 
  \end{array}     
  \end{array}
  \begin{array}[t]{r@{\ \mapsto\ }lll}
  \nu a.P &  \langle (\nuof{n}, n+1, t) \mapsto \transl{P}, (\_,\_) \mapsto \emptyset \rangle \\
  \sum_{i \in I} \alpha_i.P_i &  \langle
  \begin{array}[t]{r@{\ \mapsto\ }l}
    ((\iotaposnj, n, t) &  \sum_{k \in I_{\overline{j}}} \transl{P_k}, \\
    (\iotanegnj, n, t) &  \sum_{k \in I_{j}} \transl{P_k}\ )_{j \in n}, \\
    (\_, \_) &  \emptyset\ \ 
    \rangle.
  \end{array}\end{array}$$
Let us explain intuitions and notation.  In the first case, we assume
implicitly that $(x \colon k) \in \Xi$; the intuition is just that we
approximate variables with empty strategies.  Next, $P|Q$ is
translated to the strategy with one initial state, which only accepts
the forking move first, and then lets its avatars play $\transl{P}$
and $\transl{Q}$, respectively. In the definition, we denote by $t_1$
and $t_2$ the two players of the final position in the forking
move~\eqref{eq:para}. Furthermore, here and in all relevant cases, $n$
is the number of names in $\Gam$. Similarly, $\nu a.P$ is translated
to the strategy with one initial state, accepting only the name
creation move, and then playing $\transl{P}$. Here and in the next
case, $t$ is the player of the final position in the involved move.
In the last case, the guarded sum $\sum_{i \in I} \alpha_i.P_i$ is
translated to the strategy with one initial state, which
\begin{itemize}
\item accepts input on any channel $a$ when $\alpha_i = a$ for some $i
  \in I$, and output on any channel $a$ when $\alpha_i = \abar$ for
  some $i \in I$;
\item after an input on $a$, plays the sum of all $\transl{P_i}$'s
  such that $\alpha_i = a$; and after an output on $a$, plays the sum
  of all $\transl{P_i}$'s such that $\alpha_i = \abar$.
\end{itemize}
Formally, in the definition, we let for all $j \in n$
$I_{\overline{j}} = \ens{i \in I \aalt \alpha_i = \overline{a_j}}$ and
$I_j = \ens{i \in I \aalt \alpha_i = a_j}$.  In particular, for the
last case, when $I = \emptyset$, we obtain $0$.

Thus, almost all translations of open processes have exactly one
initial state, i.e., map the empty view on $n$ to the singleton
$1$. The only exceptions are variable applications, which are mapped
to the empty presheaf.

The translation extends to global processes as follows. Fixing a
global process $\recin{x_1 (\Del_1) = P_1, \ldots, x_k (\Del_k) =
  P_k}{P}$ typed in $\Gam$ with $n$ names, define the sequence
$(P^i)_{i \in \omega}$ of open processes (all typed in $\Xi;\Gam$) as
follows.  First, $P^0 = P$. Then, let $P^{i+1} = \partial P^i$, where
$\partial$ is the \emph{derivation} endomap on open processes typed in
any extension $\Xi; (\Gam,\Del)$ of $\Xi;\Gam$, which unfolds one
layer of recursive definitions. This map is defined by induction on
its argument as follows:
 $$\begin{array}[t]{r@{\ =\ }lll}
\partial(x_l (a_1, \ldots, a_{k_l})) & P_l \subs{b_j \mapsto a_j}_{1 \leq j \leq k_l} \\
  \partial(P|Q) & {\partial P} | {\partial Q} 
  \end{array} \hspace*{1cm}
  \begin{array}[t]{r@{\ =\ }lll}
    \partial(\nu a.P) &  \nu a. \partial P \\
    \partial(\sum_{i \in I} \alpha_i.P_i) &  \sum_{i \in I} \alpha_i.(\partial P_i),
\end{array}$$
where for all $l \in \ens{1, \ldots, k}$, $\Del_l = (b_1, \ldots,
b_{k_l})$, and $P\subs{\sigma}$ denotes simultaneous, capture-avoiding
substitution of names in $P$ by $\sigma$.

By construction, the translations of these open processes form a
sequence $\transl{P^0} \into \transl{P^1} \ldots$ of inclusions in
$\OPsh {\CV_n}$, such that for any natural number $i$ and view $V \in
\CV_n$ of length $i$ (i.e., with $i$ basic moves), $\transl{P^j} (V)$
is fixed after $j = (k+1)i$, at worst, i.e., for all $j \geq (k+1)i$,
$\transl{P^j} (V) = \transl{P^{(k+1)i}} (V)$.  Thus, this sequence has a
colimit in $\OPsh {\CV_n}$, the presheaf sending any view $V$ of
length $i$ to $\transl{P^{(k+1)i}} (V)$, which we take as the
translation of the original process.

Which equivalence is induced by this mapping on CCS, especially when
taking into account the interactive equivalences developed in the next
section? This is the main question we will try to address in future
work.

\section{Interactive equivalences}\label{sec:inter}
\subsection{Fair testing vs.\ must testing: the standard case}
An important part of concurrency theory consists in studying
\emph{behavioural equivalences}. Since each such equivalence is
supposed to define when two processes behave the same, it might seem
paradoxical to consider several of them.  Van
Glabbeek~\cite{DBLP:conf/concur/Glabbeek90} argues that each
behavioural equivalence corresponds to a physical scenario for
observing processes.

A distinction we wish to make here is between \emph{fair} scenarios,
and \emph{potentially unfair} ones. An example of a fair scenario is
when parallel composition of processes is thought of as modelling
different physical agents, e.g., in a game with several players.
Otherwise said, players are really independent. On the other hand, an
example of a potentially unfair scenario is when parallelism is
implemented via a scheduler.

Mainstream notions of processes, e.g., transition systems or automata,
are actually unfair, as the following example shows.  Consider a
looping process $\Omega$, which has a silent transition $\tau$ to
itself.  The process $P = (\Omega | \abar)$, which in parallel plays
$\Omega$ and tries to synchronise on $a$, has an infinite trace
$$P \xto{\tau} P \xto{\tau} \ldots$$

This has consequences on so-called \emph{testing}
equivalences~\cite{DBLP:journals/tcs/NicolaH84}. Let $\tick$ be a
fixed action.
\begin{defn}
  A process $P$ is \emph{must orthogonal} to a context $C$, notation
  $P \mathrel\bot^m C$, when all maximal traces of $C[P]$ play $\tick$ at some
  point.
\end{defn}
Here, maximal means either infinite or finite without extensions.
Let $P^{\bot^m}$ be the set of all contexts must orthogonal to $P$.

\begin{defn}
  $P$ and $Q$ are \emph{must equivalent}, notation $P \sim_m Q$,
  when $P^{\bot^m} = Q^{\bot^m}$.
\end{defn}

In transition systems, or automata, recalling $P$ above and letting $Q
= \Omega$, we have $P \sim_m Q.$ This might be surprising, because the
context $C = a.\tick \para \trou$ intuitively should distinguish $P$
from $Q$, by being orthogonal to $P$ but not to $Q$. However, it is
not orthogonal to $P$, because $C[P]$ has an infinite looping trace
giving priority to $\Omega$.  This looping trace is unfair, because
the synchronisation on $a$ is never performed.  Thus, one may view the
equivalence $P \sim_m Q$ as taking into account potential unfairness
of a hypothetical scheduler.  Usually, concurrency theorists consider
this too coarse, and resort to \emph{fair} testing equivalence.
  \begin{defn}
    A process $P$ is \emph{fair orthogonal} to a context $C$, notation
    $P \mathrel\bot^f C$, when all finite traces of $C[P]$ extend to traces
    that play $\tick$ at some point.
  \end{defn}
  Again, $P^{\bot^f}$ denotes the set of all contexts fair orthogonal
  to $P$.
  \begin{defn}
    $P$ and $Q$ are \emph{fair equivalent}, notation $P \sim_f Q$,
    when $P^{\bot^f} = Q^{\bot^f}$.
  \end{defn}
  This solves the issue, i.e., $P \nsim_f Q$.  

  In summary, the mainstream setting for testing equivalences relies
  on traces; and the notion of maximality for traces is intrinsically
  unfair. This is usually rectified by resorting to fair testing
  equivalence over must testing equivalence.
  Our setting is more flexible, in
  the sense that maximal plays are better behaved than maximal
  traces. In terms of the previous section, this allows viewing the
  looping trace $P \xto{\tau} P \xto{\tau} \ldots$ as non-maximal. In
  the next sections, we define an abstract notion of interactive
  equivalence (still in the particular case of CCS but in our
  setting), instantiate it to define fair and must testing
  equivalence, which, as we finally show, coincide.

\subsection{Interactive equivalences}
\begin{defn}
  A play is \emph{closed-world} when all its inputs
  and outputs are part of a synchronisation.
\end{defn}

Let $\CW \into \CE$ be the full subcategory of closed-world plays,
$\CWofX$ being the \emph{fibre} over $X$ for the projection functor
$\CW \to \Eh$, i.e., the subcategory of $\CW$ consisting of
closed-world plays with base $X$, and morphisms $(\id_X, k)$ between
them\footnote{This is not exactly equivalent to what could be noted
  $\CW_X$, since in the latter there are objects $U \otni Y \into X$
  with a strict inclusion $Y \into X$. However, both should be
  equivalent for what we do in this paper, i.e., fair and must
  equivalences.}.

Let the category of \emph{closed-world behaviours} on $X$ be the
category $\GlX = \CWofXhat$ of presheaves on $\CWofX$.  We may now
put:
  \begin{defn}
  An \emph{observable criterion} consists for all positions $X$, of a
  replete subcategory $\bbot_X \into \GlX$.
\end{defn}
Recall that $\bbot_X$ being replete means that for all $F \in \bbot_X$
and isomorphism $f \colon F \to F'$ in $\GlX$, $F'$ and $f$ are in
$\bbot_X$.  

An observable criterion specifies the class of `successful',
closed-world behaviours. The two criteria considered below are two
ways of formalising the idea that a successful behaviour is one in
which all accepted closed-world plays are `successful', in the sense
that some player plays the tick move at some point.

We now define interactive equivalences.  Recall that $[F, G]$ denotes
the amalgamation of $F$ and $G$, and that right Kan extension along
$\op{i_Z}$ induces a functor $\Ran_{\op{i_Z}} \colon \widehat{\CV_Z}
\to \widehat{\CE_Z}$. Furthermore, precomposition with the canonical
inclusion $j_Z \colon \CW (Z) \into \CE_Z$ induces a functor $j_Z^*
\colon \widehat{\CE_Z} \to \widehat{\CW (Z)}$.  Composing the two, we
obtain a functor $\SStoGG \colon \SS_Z \to \Gl_Z$:
$$\SS_Z = \widehat{\CV_Z} \xto{\Ran_{\op{i_Z}}} \widehat{\CE_Z} \xto{j_Z^*} \widehat{\CW (Z)} = \Gl_Z.$$

\noindent \begin{minipage}[c]{0.7\linewidth}
\begin{defn}\label{defn:orth}
  For any strategy $F$ on $X$ and any pushout square $P$ of positions
  as on the right, with $I$ consisting only of channels, let
  $F^{\bbot_P}$ be the class of all strategies $G$ on $Y$ such that
  $\SStoGG([F, G]) \in \bbot_Z$.
\end{defn}
  \end{minipage}
  \begin{minipage}[c]{0.28\linewidth}
  \begin{equation}
    \Diag(.3,.8){%
      \pbk{X}{Z}{Y} %
    }{%
      |(I)| I \& |(Y)| Y \\
      |(X)| X \& |(Z)| Z %
    }{%
      (I) edge (Y) %
      edge (X) %
      (X) edge (Z) %
      (Y) edge (Z) %
    }\label{eq:orth}
  \end{equation}
  \end{minipage}

  Here, $G$ is thought of as a \emph{test} for $F$. Also, $P$ denotes
  the whole pushout square and $F^{\bbot_P}$ is notation for a notion
  indexed by such squares, whose definition uses $\bbot_Z \into
  \Gl_Z$.  From the CCS point of view, $I$ corresponds to the set of
  names shared by the process under observation $(F)$ and the testing
  context $(G)$.

\begin{defn} Any two strategies $F,F'\in\SSX$ are
  \emph{$\bbot$-equivalent}, notation $F \sim_{\bbot} F'$, iff for all
  pushouts $P$ as in~\ref{eq:orth}, $F^{\bbot_P} = {F'}^{\bbot_P}$.
\end{defn}

\subsection{Fair vs.\ must}\label{subsec:fairmust}
Let us now define fair and must testing equivalences.  Let a
closed-world play be \emph{successful} when it contains a $\tickn$.
Furthermore, for any closed-world behaviour $G \in \GlX$ and
closed-world play $U \in \CWofX$, an \emph{extension} of a state
$\state \in G (U)$ to $U'$ is a $\statei \in G (U')$ with $i \colon U
\to U'$ and $G (i)(\statei) = \state$. The extension $\statei$ is
\emph{successful} when $U'$ is. The intuition is that the behaviour
$G$, before reaching $U'$ with state $\sigma'$, passed through $U$
with state $\sigma$.
  \begin{defn}
    The \emph{fair} criterion $\bbot^f$ contains all closed-world
    behaviours $G$ such that any state $\state \in G (U)$ for finite $U$
    admits a successful extension.
  \end{defn}

  Now call an extension of $\state \in G (U)$ \emph{strict} when $U
  \to U'$ is not surjective, or, equivalently, when $U'$ contains more
  moves than $U$.  For any closed-world behaviour $G \in \GlX$, a state
  $\state \in G (U)$ is $G$-\emph{maximal} when it has no strict
  extension.

  \begin{defn}
  Let the \emph{must} criterion $\bbot^m$ consist of all closed-world behaviours
  $G$ such that for all closed-world $U$ and $G$-maximal $\state \in G
  (U)$, $U$ is successful.
\end{defn}

We now show that fair and must testing equivalence coincide. The key
result for this is:
\begin{thm}\label{thm:fairmust}
  For any strategy $F$ on $X$, any state $\state \in \SStoGG(F) (U)$
  with finite $U$ admits a $\SStoGG(F)$-maximal extension.
\end{thm}
The proof basically amounts to implementing a scheduler in our
framework --- a fair one, of course.  Thanks to the theorem, we have:
  \begin{lem}\label{lem:fairmust}
    For all $F \in \SS_X$, $\SStoGG(F) \in \bbot^m_X$ ~iff~ $\SStoGG(F)
    \in \bbot^f_X$.
\end{lem}
\begin{proof}
  Let $G = \SStoGG (F)$.

  \emph{($\Rightarrow$)} By Theorem~\ref{thm:fairmust}, any state
  $\state \in G (U)$ has a $G$-maximal extension $\statei \in G (U')$,
  which is successful by hypothesis, hence $\state$ has a successful
  extension.

\emph{($\Leftarrow$)} Any $G$-maximal $\state \in G (U)$ admits by hypothesis
a successful extension which may only be on $U$ by $G$-maximality, and
hence $U$ is successful.
\end{proof}

Now comes the expected result:
  \begin{thm}
For all $F, F' \in \SS_X$, $F \sim_{\bbot^m} F'$ iff $F \sim_{\bbot^f} F'$.
\end{thm}
\begin{proof}  
  \emph{($\Rightarrow$)} Consider two strategies $F$ and $F'$ on $X$,
  and a strategy $G$ on $Y$ (as in the pushout~\eqref{eq:orth}). We
  have, using Lemma~\ref{lem:fairmust}:
  \begin{center}
    $\SStoGG(F \glue G) \in \bbot^f$ \hfil iff \hfil $\SStoGG(F \glue G) \in \bbot^m$ \hfil iff \hfil $\SStoGG(F'
    \glue G) \in \bbot^m$ \hfil iff \hfil $\SStoGG(F' \glue G) \in \bbot^f$.
  \end{center}
    \emph{($\Leftarrow$)}  Symmetric.
  \end{proof}

  To explain what is going on here, let us consider again $P = (\Omega
  | \abar)$, $Q = \Omega$, and the context $C = a.\tick \para
  \trou$. We implement $C$ by choosing as a test the strategy $T =
  \transl{a.\tick}$ on a single player knowing one name $a$. Taking
  $I$ to consist of the sole name $a$, the pushout $Z$ as in
  Definition~\ref{defn:orth} consists of two players, say $x$ for the
  observed strategy and $y$ for the test strategy, sharing the name
  $a$.  Now, assuming that $\Omega$ loops deterministically, the
  global behaviour $G = \SStoGG([\transl{P}, T])$ has exactly one state on
  the empty play, and again exactly one state on the play $\paraof{1}$
  consisting of only one fork move by $x$. Thus, $G$ reaches a
  position with three players, say $x_1$ playing $\Omega$, $x_2$
  playing $\abar$, and $y$ playing $a.\tick$.  What makes the theorem
  work is that the play with $\omega$ silent moves by $x_1$ is not
  maximal.  Indeed, we could insert (anywhere in the sequence of moves
  by $x_1$) a synchronisation move by $x_2$ and $y$, and then a tick
  move by the avatar of $y$. Essentially: our notion of play is more
  fair than just traces.

\appendix

\section{Diagrams}\label{sec:diagrams}
In this section, we define the category on which our diagrams are
presheaves. The techniques used here date back at least to Carboni and
Johnstone~\cite{DBLP:journals/mscs/CarboniJ95,DBLP:journals/mscs/Johnstone04}.
Let us first consider two baby examples. It is well-known that
directed multigraphs form a presheaf category: consider the category
$\C$ freely generated by the graph with two vertices, say $\star$ and
$[1]$, and two edges $d, c \colon \star \to [1]$ between them.  One
way to visualise this is to compute the \emph{category of elements} of
a few presheaves on $\C$. Recall that the category of elements of a
presheaf $F$ on $\C$ is the comma category $y \downarrow_{\Chat} F$,
where $y$ is the Yoneda embedding. Via Yoneda, it has as elements
pairs $(C, x)$ with $C \in \ob (\C)$ and $x \in F (C)$, and morphisms
$(C,x) \to (D,y)$ morphisms $f \colon C \to D$ in $\C$ such that $F
(f) (y) = x$ (which we abbreviate as $y \cdot f = x$ when the context
is clear).
\begin{ex}
  Consider the presheaf $F$ defined by the following equations, whose
  category of elements is actually freely generated by the graph on
  the right:
\begin{center}
    \begin{minipage}[c]{0.3\linewidth}
      \begin{itemize}
      \item $F (\star) = \ens{0,1,2}$,
      \item $F ([1]) = \ens{e, e'}$,
      \end{itemize}
    \end{minipage}
    \hfil
    \begin{minipage}[c]{0.2\linewidth}
      \begin{itemize}
      \item $e \cdot d = 0$,
      \item $e \cdot c = 1$,
      \item $e' \cdot d = 1$,
      \item $e' \cdot c = 2$,
      \end{itemize}
    \end{minipage}
      \hfil
    \diag(.15,.8){%
      \& \& |(un)| 1 \& \\ %
      \& |(e)| e \& \& |(e')| e' \\ %
      |(zer)| 0 \& \& \& \& |(deux)| 2. %
    }{%
      (zer) edge[labelal={d}] (e) %
      (un) edge[labelal={c}] (e) %
      edge[labelar={d}] (e') %
      (deux) edge[labelar={c}] (e') %
    }
\end{center}
This graph is of course not exactly the expected one, but it does
represent it. Indeed, for each vertex we know whether it is in $F
(\star)$ or $F ([1])$, hence whether it represents a `vertex' or an
`edge'. The arrows all go from a `vertex' $v$ to an `edge' $e$. They
are in $F (d)$ when $v$ is the domain of $e$, and in $F (c)$ when $v$
is the codomain of $e$.
\end{ex}

Multigraphs may also be seen as a presheaves on the category freely
generated by the graph with
\begin{itemize}
\item as vertices: one special vertex $\star$, plus for each natural
  number $n$ a vertex, say, $[n]$; and
\item $n+1$ edges $\star \to [n]$, say $d_1, \ldots, d_n$, and $c$.
\end{itemize}
It should be natural for presheaves on this category to look like
multigraphs: the elements of a presheaf $F$ above $\star$ are the
vertices in the multigraph, the elements above $[n]$ are the $n$-ary
multiedges, and the action of the $d_i$'s give the $i$th source of a
multiedge, while the action of $c$ gives its target.
\begin{ex}
  Similarly, computing a few categories of elements might help
  visualising. As above, consider $F$ defined by
\begin{center}
    \begin{minipage}[c]{0.3\linewidth}
      \begin{itemize}
      \item $F (\star) = \ens{0,1,2,3,4}$,
      \item $F ([1]) = F ([0]) = \emptyset$,
      \item $F ([2]) = \ens{e'}$,
      \item $F ([3]) = \ens{e}$,
      \item $F ([n+4]) = \emptyset$,
      \end{itemize}
    \end{minipage}
    \hfil
    \begin{minipage}[c]{0.2\linewidth}
      \begin{itemize}
      \item $e \cdot c = 0$,
      \item $e \cdot d_1 = 1$,
      \item $e \cdot d_2 = 2$,
      \item $e \cdot d_3 = 3$,
      \end{itemize}
    \end{minipage}
    \hfil
    \begin{minipage}[c]{0.2\linewidth}
      \begin{itemize}
      \item $e' \cdot c = 1$,
      \item $e' \cdot d_1 = 4$,
      \item $e' \cdot d_2 = 5$,
      \end{itemize}
    \end{minipage}
\end{center}
whose category of elements is freely generated by the graph:
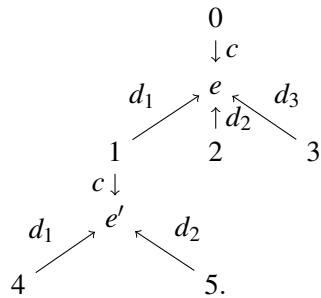
\begin{center}
    \diag(.3,.8){%
      \& \& |(zero)| 0 \& \\ %
      \& \& |(e)| e \& \\ %
      \& |(un)| 1 \& |(deux)| 2 \& |(trois)| 3 \\ %
      \& |(e')| e' \& \\ %
      |(quatre)| 4 \& \& |(cinq)| 5. %
    }{%
      (zero) edge[labelr={c}] (e) %
      (un) edge[labelal={d_1}] (e) %
      edge[labell={c}] (e') %
      (deux) edge[labelr={d_2}] (e) %
      (trois) edge[labelar={d_3}] (e) %
      (quatre) edge[labelal={d_1}] (e') %
      (cinq) edge[labelar={d_2}] (e') %
    }
\end{center}
\end{ex}

Now, this pattern may be extended to higher dimensions. Consider for
example extending the previous base graph with a vertex $[m_1,
\ldots,m_n; p]$ for all natural numbers $n, p, m_1, \ldots, m_n$, plus
edges $s_1\colon [m_1] \to [m_1, \ldots, m_n; p], \ldots, s_n\colon
[m_n] \to [m_1, \ldots, m_n;p]$, and $t \colon [p] \to [m_1, \ldots,
m_n;p]$. Let now $\C$ be the free category on this extended
graph. Presheaves on $\C$ are a kind of 2-multigraphs: they have
vertices, multiedges, and multiedges between multiedges.

We could continue this in higher dimensions.

Defining the base category of the paper follows a very similar
pattern. We start from a slightly different graph: let $\G_0$ have
just one vertex $\star$; let $\G_1$, have one vertex $\star$, plus a
vertex $[n]$ for each natural number $n$, plus $n$ edges $d_1, \ldots,
d_n \colon \star \to [n]$.  Let $\C_0$ and $\C_1$ be the categories
freely generated by $\G_0$ and $\G_1$, respectively. So, presheaves on
$\C_1$ are a kind of hypergraphs with arity (since vertices incident
to a hyperedge are numbered). This is enough to model positions.

Now, consider the graph $\G_2$, which is $\G_1$ augmented with:
\begin{itemize}
\item for all $n$, vertices $\tickn$, $\paraln$, $\pararn$, $\nun$, 
\item for all $n$ and $0 \leq i < n$, vertices $\iotaposni$ and $\iotanegni$,
\item for all $n$, edges $s, t \colon [n] \to \tickn$, $s, t \colon [n] \to \paraln$, $s, t \colon [n] \to \pararn$, 
$s \colon [n] \to \nun$, $t \colon [n+1] \to \nun$, 
\item for all $n$ and $0 \leq i < n$, edges $s, t \colon [n] \to \iotaposni$, $s,t \colon [n] \to \iotanegni$.
\end{itemize}
Note that only name creation changes the number of names known to the
player, and accordingly the corresponding morphism $t$ has domain
$[n+1]$. We slightly abuse language here by calling all these $t$'s
and $s$'s the same. We could label them with their codomain, but we
refrain doing so for the sake of readability.

Now, let $\C_2$ be the category generated by $\G_2$ and the relations
$s \rond d_i = t \rond d_i$ for all $n$ and $0 \leq i < n$ (for all
possible --- common --- codomains with main index $n$ for $s$ and
$t$).  Presheaves on $\C_2$ are enough to model views, but since we
want more, we continue, as follows.

Let $\G_3$ be $\G_2$, augmented with:
\begin{itemize}
\item for all $n$, a vertex $\paran$, and 
\item edges $l \colon \paraln \to \paran$ and $r \colon \pararn \to \paran$.
\end{itemize}
Let $\C_3$ be the category generated by $\G_3$ and the relations $l
\rond s = r \rond s$ (this models the fact that a forking move should
be played by just one player).  Presheaves on $\C_3$ are enough to
model full moves; to model closed-world moves, and in particular
synchronisation, we continue as follows.

Let $\G_4$ be $\G_3$, augmented with, for all $n$, $m$, $0 \leq i <
n$, and $0 \leq j < m$,
 \begin{itemize}
 \item a vertex $\taunimj$, and
 \item edges $\sender \colon \iotaposni \to \taunimj$ and $\receiver
   \colon \iotanegmj \to \taunimj$.
 \end{itemize}
 Let $\C_4$ be the category generated by $\G_4$ and the relations
 $\sender \rond s \rond d_i = \receiver \rond s \rond d_j$ (which
 models the fact that a synchronisation involves an input and an
 output on the same name).

\bigskip
\paragraph{Acknowledgements} Thanks to the courageous having endured
the first versions of this work. Special thanks to Paul-Andr\'e
Melli\`es for his graphical design skills, and to pseudonymous referee
Michel Houellebecq, not only for our very useful and enjoyable
discussion, but also for tolerating our rather poor litterary
style. Finally, thanks to Mark Weber for teaching the first author the
techniques of Appendix~\ref{sec:diagrams}.

\bibliographystyle{eptcs} \bibliography{bib}
\end{document}